\documentclass[runningheads]{llncs}
\usepackage{graphicx}
\usepackage{algorithmic}
\usepackage[ruled,linesnumbered]{algorithm2e}
\usepackage{textcomp}
\usepackage{hyperref}       % hyperlinks
\usepackage{url}            % simple URL typesetting
\usepackage{booktabs}       % professional-quality tables
\usepackage{nicefrac}
\usepackage{microtype}      % microtypography
\usepackage{changes}
\usepackage{lipsum}
\usepackage{enumitem}
\usepackage{listings}
\usepackage{array}
\usepackage{subcaption}
\usepackage{wrapfig}
\usepackage{mathtools,cuted}
\usepackage{color}
\usepackage{colortbl}
\usepackage{cases}
\usepackage{amsmath}
\usepackage{float}

\DeclarePairedDelimiter\norm{\lVert}{\rVert}
\usepackage{stackengine,graphicx}
\stackMath

\usepackage{amssymb}
\usepackage[normalem]{ulem}
\usepackage{bm}
\usepackage{mathtools}
\usepackage{mathrsfs}
\usepackage{verbatim}
\usepackage{comment}

\newcommand{\xbf}{\mathbf{x}}
\newcommand{\ybf}{\mathbf{y}}

\newcommand{\cbf}{\mathbf{c}}
\newcommand{\fbf}{\mathbf{f}}

\newcommand{\Xbf}{\mathbf{X}}

\newcommand{\etol}{\epsilon_{\mathrm{tol}}}

\newcommand{\CC}{\mathbb{C}}
\DeclareMathOperator*{\argmin}{arg\,min}

\newcommand{\D}{\mathcal{D}}

\newcommand{\xt}{\mathbf{X}^{(t)}}

\newcommand{\xtp}{\mathbf{X}^{(t+1)}}

\newcommand{\Scal}{\mathcal{S}}

\newcommand{\epst}{\varepsilon_{t}}

\newcommand{\epstp}{\varepsilon_{t+1}}

\usepackage{tikz}
\usetikzlibrary{shapes,arrows,chains}

\newcommand{\Chen}[1]{\textcolor{red}{[Chen: #1]}}

\begin{document}
	\title{A Learnable Variational Model for Joint Multimodal MRI Reconstruction and Synthesis\thanks{This work was supported in part by National Science Foundation under grants DMS-1818886, DMS-1925263, DMS-2152960 and DMS-2152961.}}
	%\title{Contribution Title\thanks{Supported by organization x.}}
	%
	\titlerunning{Learnable Variational Model}
	\author{Wanyu Bian\inst{1} \and 
		Qingchao Zhang\inst{1} \and
		Xiaojing Ye\inst{2} \and
		Yunmei Chen\inst{1}}
	% 1{Bian, Wanyu} 2{Qingchao, Zhang} 3{Ye, Xiaojing} 4{Chen, Yunmei}
	%
	\authorrunning{W. Bian et al.}
	\institute{University of Florida, Gainesville FL 32611, USA\\
		\email{ \{wanyu.bian, qingchaozhang, yun\}@ufl.edu} \and
		Georgia State University, Atlanta GA 30302, USA\\
		\email{xye@gsu.edu}}
	\maketitle              % typeset the header of the contribution
	\begin{abstract}
		Generating multi-contrasts/modal MRI of the same anatomy enriches diagnostic information but is limited in practice due to excessive data acquisition time. In this paper, we propose a novel deep-learning model for joint reconstruction and synthesis of multi-modal MRI using incomplete k-space data of several source modalities as inputs. The output of our model includes reconstructed images of the source modalities and high-quality image synthesized in the target modality.
		Our proposed model is formulated as a variational problem that leverages several learnable modality-specific feature extractors and a multimodal synthesis module. We propose a learnable optimization algorithm to solve this model, which induces a multi-phase network whose parameters can be trained using multi-modal MRI data. Moreover, a bilevel-optimization framework is employed for robust parameter training. We demonstrate the effectiveness of our approach using extensive numerical experiments.
		
		\keywords{MRI Reconstruction  \and Multimodal MRI Synthesis  \and Deep Neural Network \and Bilevel-Optimization.}
	\end{abstract}
	\section{Introduction}
	Magnetic resonance imaging (MRI) is a prominent leading-edge medical imaging technology which provides diverse image contrasts under the same anatomy. Multiple different contrast images are generated by varying the acquisition parameters, e.g. T1-weighted (T1), T2-weighted (T2) and Fluid Attenuated Inverseion Recovery (FLAIR). They have similar anatomical structure but highlight different soft tissue  which enriches the diagnostic information for clinical applications and research studies comparing to single modality \cite{iglesias2013synthesizing}.  
	%Different contrasts of MR images are also referred as different modalities. 
	%For example, acquiring  images uses short TR and TE times, high fat content appear bright and compartments filled with  celebrospinal fluid (CSF) appear dark and they provide more anatomical information. T2-weighted (T2) images require longer TR and TE times. In general, T2-weighted images appear to be a reversal of T1-weighted images in contrast and they often provide more pathological information for delineation of edema. Fluid Attenuated Inverseion Recovery (FLAIR) need very long TR and TE times comparing to T1 and T2, their contrast appear similar as T2-weighted images with supressed CSF so that the free water becomes dark and minimizes contrast between gray matter and white matter.  Multimodal MR images can provide more diagnostic information for clinical application and research studies comparing to single modality \cite{iglesias2013synthesizing}. 
	A major limitation of MRI is its relatively long data acquisition time during MRI scans. It does not only cause patient discomfort, but also makes MR images prone to motion artifacts which degrade the diagnostic accessibility. A mainstream routine to reduce the MRI acquisition time is to reconstruct partially undersampled k-space acquisitions, another approach is to synthesize target modality MR image from fully-sampled acquisitions of source modality images \cite{yang2020model,dar2020prior}.  
	%\ye{It feels that we need to explain why we want to do reconstruction and synthesis jointly.}
	
	Compressed sensing MRI (CS-MRI) reconstruction is a predominant approach for accelerating MR acquisitions, which solves an inverse problem formulated as a variational model.  In recent decades, deep learning based models have leveraged large datasets and further explored the potential improvement of reconstruction performance. Most of the deep learning based reconstruction methods employ end-to-end deep networks \cite{yang2017dagan,sriram2020end,lee2018deep}. To overcome the weakness of the end-to-end black-box networks,  several learnable optimization
	algorithms (LOAs) have been developed attracted much attention, which possess of a more interpretable network architecture. LOA-based reconstruction methods unfold the iterative optimization algorithm into a multi-phase network in which the regularization and image transformation are learned to improve the network performance \cite{monga2021algorithm,jimaging7110231,chen2021learnable,bian2022optimal,pmlr-v119-bao20b}, e.g. ADMM-Net \cite{sun2016deep}, ISTA-Net$^+$ \cite{zhang2018ista} and PD-Net \cite{cheng2019model}. 
	
	Multimodal MR image synthesis in recent years has emerged using various deep learning frameworks, where a main stream is to start with source modalities from the image domain and synthesize the images of the target modalities \cite{dar2019image,sohail2019unpaired,welander2018generative}. For instance, GAN-based methods are mostly end-to-end from images to images with encoder-decoder architectures in their generator networks. MM-GAN \cite{sharma2019missing} channel-wisely concatenates all the available modalities with a zero image for missing modality and imputes the missing input incorporating curriculum learning for GAN.  Multimodal MR (MM) \cite{liu2020multimodal}, MMGradAdv \cite{chartsias2017multimodal} and Hi-Net \cite{zhou2020hi} exploit the correlations between multimodal data and apply robust feature fusion method to form a unified latent representation. A rarely explored approach \cite{dar2020prior} starts from  undersampled  k-space data of the source modalities to generate target modality images. This paper further explored this approach  and the major differences from \cite{dar2020prior} to ours are: (i) In \cite{dar2020prior} it requires that the target modality is heavily undersampled and the source modality is lightly undersampled, while our method does not require any of the k-space information of target modality nor the source modality to be lightly undersampled which is much less limitations in real-world applications; (ii) Instead of learning the mapping from image to image, we learn the mapping from the features of source images to the target image since features provide more direct information to synthesize images of a new modality and (iii) We formulate the joint reconstruction and synthesis problem in a variational model and propose a convergent algorithm as the architecture of the deep neural network so that the network is interpretable and convergent.

	In order to synthesize target modality by using partially scanned k-space data from source modalities in stead of fully scanned data that used in the state-of-the-art multimodal synthesis, in this paper, we propose to jointly reconstruct undersampled multiple source modality MR images and synthesize the target modality image. Our contributions are summarized as follows: (1) We propose a novel LOA for joint multimodal MRI reconstruction and synthesis with theoretical analysis guarantee;
	(2) The parameters and hyper-parameters of the network induced by our LOA are learned using a bilevel optimization algorithm robust parameter training;
	(3) Extensive experimental results demonstrate the efficiency of the proposed method and high quality of the reconstructed/synthesized images.

	We demonstrate that our proposed joint synthesis-reconstruction network can further improve image reconstruction quality over existing sole reconstruction networks using the same partial k-space measurements. This improvement is due to the additional image feature information provided by the synthesis functionality of our network, which is trained by comprehensive image data of all relevant modalities together. Moreover, the synthesized images can serve as additional references to radiologists when the corresponding k-space data cannot be acquired in practice.

	%\section{Related Work}
	\section{Proposed Method}
	\subsection{Model}
	In this section, we provide the details of the proposed model for joint MRI reconstruction and synthesis. Given the partial k-space data $\{ \fbf_1, \fbf_2 \}$ of the source modalities (e.g. T1 and T2), our goal is to reconstruct the corresponding images $\{\xbf_1, \xbf_2\}$ \emph{and} synthesize the image $\xbf_3$ of the missing modality (e.g. FLAIR) without its k-space data. 
	To this end, we propose to learn three modality-specific feature extraction operators $\{h_{w_i}\}_{i=1}^{3}$, one for each of these three modalities. Then, we design regularizers of these images by combining these learned operators and a robust sparse feature selection operator (we use $(2,1)$-norm in this work). 
	To synthesize the image $\xbf_3$ using $\xbf_1$ and $\xbf_2$, we employ another feature-fusion operator $g_{\theta}$ which learns the mapping from the features $h_{w_1}(\xbf_1)$ and $h_{w_2}(\xbf_2)$ to the image $\xbf_3$.
	Our proposed model reads as
	%%
	%Our model takes the down-sampled k-space data $\{ f_i \} _{i = 1}^{N}$ as input, and reconstructs the corresponding images $\{ x_i \} _{i = 1}^{N}$ as well as synthesizing the modalities $\{ x_j \} _{j = N+1}^{M}$ without any k-space data. We perform the joint reconstruction and synthesis by minimizing the following function
	%
	\begin{equation}
		\label{our_model}
		\begin{aligned}
			\min_{\xbf_1, \xbf_2, \xbf_3}\Psi_{\Theta, \gamma}(\xbf_1, \xbf_2, \xbf_3) & :=  \textstyle{\frac{1}{2} \sum\limits_{i = 1}^{2}  \| P_i F \xbf_i  - \fbf_i\|_2^2 + \frac{1}{3} \sum\limits_{i = 1}^{3}   \|h_{w_i} (\xbf_i) \|_{2,1}}  \\
			&\quad + \textstyle{\frac{\gamma}{2} \|g_{\theta} ([h_{w_1}(\xbf_1),h_{ w_2} (\xbf_2)]) - \xbf_3\|_2^2},
		\end{aligned}
	\end{equation}
	where the first term in \eqref{our_model} is the data fidelity for the source modalities that ensures consistency between the reconstructed images $\{\xbf_1, \xbf_2\}$ and the sensed partial k-space data $\{ \fbf_1, \fbf_2 \}$. Here, $F$ stands for the discrete Fourier transform and $P_i$ is the binary matrix representing the k-space mask when acquiring data for $\xbf_i$. In \eqref{our_model}, $h_{w_i}$ represents the modality-specific feature extraction operator which maps the input $\xbf_i \in \CC^n$ to a high-dimensional feature tensor $h_{w_i} (\xbf_i) \in \CC^{m \times d}$, where $m$ is the spatial dimension and $d$ is the channel number of the feature tensor. The second term is the regularization of all modalities $\{\xbf_1, \xbf_2, \xbf_3\}$ to enhance sparsity of the their feature tensors, where 
	$\label{eq:r}
	\|h_{w_i}(\xbf_i)\|_{2,1} = \sum_{j = 1}^{m} \|h_{w_i, j}(\xbf_i)\|.
	$
	Here each $h_{w_i, j}\in \mathbb{R}^d$ can be viewed as a feature vector at spatial location $j$.
	The last term in \eqref{our_model} is to synthesize $\xbf_3$ by learning a mapping $g_{\theta} : \CC^{m \times 2d} \rightarrow \CC^{n}$ that maps the concatenated features of $\xbf_1$ and $\xbf_2$ (i.e. $[h_{w_1}(\xbf_1),h_{ w_2} (\xbf_2)]$) to
	$\xbf_3$, and  $[\cdot,\cdot]$ represents the concatenation of the arguments.
	Here $g_{\theta}$ maps features of $x_1$ and $x_2$ to the target  image so that more useful information can be used to generate the target image comparing to the mappings from source images to the target image.

	In our implementation, we parameterize the modality-specific feature extraction operator $h_{w_i}$ and the synthesis mapping $g_{\theta}$ as vanilla CNNs with $l$ and  $l'$ layers respectively, both of which use the smoothed rectified linear unit \cite{chen2021learnable} as activation. 
	For notation simplicity, we let $\Theta$ in \eqref{our_model} denote the collection of all parameters in the convolution operators of the function $\Psi$, i.e. $\Theta = \{w_1, w_2, w_3, \theta\} $. 
	
	The weight $\gamma$ is a hyper-parameter which plays a critical role in balancing the reconstruction part (first two terms in \eqref{our_model}) and the image synthesis part (last term in \eqref{our_model}) of the model \eqref{our_model}, and hence has significant impact to the final image reconstruction and synthesis quality. To address this important issue, we propose to use a bi-level hyper-parameter tuning framework to learn $\gamma$ by minimizing the reconstruction loss on both validation and training data sets. Details of this hyper-parameter tuning will be provided in Section \ref{section: bilevel}.
	\subsection{Efficient Learnable Optimization Algorithm}
	\label{loa induced net}
	%\ye{I suggest we start this subsection with a paragraph indicating that our plan is to design an optimization algorithm for \eqref{our_model} with guaranteed convergence, then unroll this algorithm to get a DNN so that the parameters can be learned. Finally the DNN inherits the convergence property. Then start the next paragraph about smoothing etc.} 
	
	In this section, we present a novel and efficient learnable optimization algorithm (LOA) for solving the nonconvex nonsmooth minimization problem \eqref{our_model}. (Comprehensive convergence analysis of this algorithm is provided in Supplementary Material.) Then we design a DNN whose architecture exactly follows this algorithm, and the parameters of the DNN can be learned from data. In this way, the DNN inherits all the convergence properties of the LOA.
	
	Since the second sum in the minimization problem \eqref{our_model} is nonsmooth due to the $l_{2, 1}$-norm, we first approximate these nonsmooth terms using their smooth surrogates $
	\|h_{w_i}(\xbf_i)\|_{\varepsilon_{2,1}} = \sum\nolimits^m_{j=1} \big( \sqrt{\| h_{w_i, j}(\xbf_i) \|^2 + \varepsilon^2} -\varepsilon \big)
	$, where $\varepsilon > 0$ is the parameter representing the smoothing level. 
	Thus, for every fixed $\varepsilon$, we obtain a smooth surrogate function $\Psi_{\Theta,\gamma}^{\varepsilon}$ of the nonsmooth objective $\Psi_{\Theta,\gamma}$, and we can apply a gradient descent step to update our approximation to the solution of \eqref{our_model}.
	In our algorithm, the smoothing level $\varepsilon$ is automatically reduced and tends to 0, such that the surrogate approaches the original nonsmooth regularizers in \eqref{our_model}.
	More precisely, let $\Xbf = \{\xbf_1, \xbf_2, \xbf_3 \}$ for notation simplicity, then we solve the problem $\min_{\Xbf}\Psi_{\Theta, \gamma}(\Xbf)$ with initial $\Xbf^{(0)}$ using Algorithm \ref{alg:lda} (the initial $\Xbf^{(0)}$ is obtained from a pre-trained initial network, which is illustrated in detail in Section \ref{Initialization}). At Line \ref{line_search} of Algorithm \ref{alg:lda}, we compute a gradient descent update with step size obtained by line search while the smoothing parameter $\epst > 0$ is fixed. In Line \ref{reduction_cre}, the algorithm updates $\epst$ based on a reduction criterion.
	The reduction of $\epst$ ensures that the specified subsequence (the subsequence who met the $\epst$ reduction criterion) must have an accumulation point that is a Clarke stationary point \cite{chen2021learnable} of the optimization problem \eqref{our_model}, as given in Theorem \ref{theorem}, whose proof is provided in Supplementary Materials. %The proof of this theorem is very similar to those given in \cite{chen2021learnable,jimaging7110231}, and hence is omitted due to space limitation.
	%
	%\blue{For network, given the prescriped the algorithm operates for $\hat{T}$ iterations. To learn the parameters of the feature extraction and image synthesis operators, we create a DNN of $\hat{T}$ phases whose structure exactly follows Algorithm \ref{alg:lda}, where each phase of the network performs one iteration of the algorithm.
	%}
	We create a multi-phase network whose architecture exactly follows Algorithm \ref{alg:lda} with a prescribed phase number $\hat{T}$, where each phase of the network performs one iteration of the algorithm.
	%\Chen{For a network, whose structure exactly follows Algorithm \ref{alg:lda}, we prescribed the number of the phases being $\hat{T}$.}
	\begin{theorem}
		Suppose that $\{\xt \}$ is the sequence generated by Algorithm \ref{alg:lda} with any initial $\Xbf^{(0)}$, $\etol=0$ and $T=\infty$. Let $ \{ \Xbf^{(t_l+1)}\}$ be the subsequence that satisfies the reduction criterion  in step \ref{reduction_cre} of Algorithm \ref{alg:lda}. Then $ \{ \Xbf^{(t_l+1)}\}$ has at least one accumulation point, and every accumulation point of $\{ \Xbf^{(t_l+1)}\}$ is a Clarke stationary point of $ \min_{\Xbf} \Psi_{\Theta, \gamma}(\Xbf)$.
		\label{theorem}
	\end{theorem}
	\begin{algorithm}[t]
		\caption{Learnable Descent Algorithm}
		\label{alg:lda}
		\begin{algorithmic}[1]
			\STATE \textbf{Input:} $\Xbf^{(0)}$, $0<\eta<1$, and $\varepsilon_0$, $a, \sigma >0$, $t = 0$. Max $T$, tolerance $\etol>0$.
			\FOR{$t=0,1,2,\dots,T-1$}
			%\WHILE{$t < T$ and $\sigma {\epst} \ge \etol$}
			\STATE $\xtp = \xt - \alpha_{t}  \nabla \Psi_{\Theta,\gamma}^{\epst} (\xt)$, where the step size $\alpha_{t}$ is obtained through \\ line search s.t. $ \Psi_{\Theta,\gamma}^{\epst}(\xtp) - \Psi_{\Theta,\gamma}^{\epst}(\xt) \le - \frac{1}{a} \| \xtp - \xt\|^2$ holds. \label{line_search}
			\STATE \textbf{if} $\|\nabla \Psi_{\Theta,\gamma}^{\epst}(\xtp)\| < \sigma \eta {\epst}$,  set $\epstp= \eta {\epst}$;  \textbf{otherwise}, set $\epstp={\epst}$. \label{reduction_cre}
			\STATE \textbf{if} $\sigma {\epst} < \etol$, 
			%\textbf{terminate} and \textbf{output} $\Xbf^{(t)}$.
			\textbf{terminate} and go to Line \ref{endfor},
			%\ENDWHILE{ and \textbf{output} $\Xbf^{(T)}$.}
			\ENDFOR{ and \textbf{output} $\Xbf^{(t)}$.} \label{endfor}
			%\STATE \textbf{Output:} $\Xbf^{(T)}$.
		\end{algorithmic}
	\end{algorithm}
	%The $\Xbf^{(T)}$ is the output of 
	%
	%This multi-phase network has $\{\Theta, \gamma\}$ as learnable parameters.
	%It takes $\fbf_1, \fbf_2$ and the initialization $\Xbf^{(0)}$ and outputs $\Xbf^{(T)} = \Xbf^{(T)}(\Theta, \gamma; \Xbf^{(0)}, \fbf_1, \fbf_2) $. To avoid ambiguity, in the following text, we let $\Xbf^{(T)}$ denote the output and $\Xbf^{(T)}(\cdot)$ denote the mapping from input to output rendered by this network.
	\subsection{Bilevel Optimization Algorithm for Network Training}
	\label{section: bilevel}
	%Given the initialization $\xbf_0$ represents the initialization of all modalities and the the partial k-space data $\{ \fbf_1, \fbf_2 \}$, we 
	Suppose that we randomly sample $\mathcal{M}_{tr}$ data pairs $\{\D^{tr}_{i} \}_{i = 1} ^{\mathcal{M}_{tr}}$ for training and $\mathcal{M}_{val}$ data pairs $\{\D^{val}_{i} \}_{i = 1} ^{\mathcal{M}_{val}}$ for validation, where each $\D^{tr}_{i}(\mbox{or} \ \D^{val}_{i}) $ is composed of data pair $ \{(\fbf_1^{i}, \fbf_2^{i}), \Xbf^{*i} \}$, $\fbf_1^{i}, \fbf_2^{i}$ denote the given partial k-space data, and $\Xbf^{*i} = \{\xbf_1^{*i}, \xbf_2^{*i}, \xbf_3^{*i} \}$ denotes the corresponding reference images.
	
	To find the optimal value of the important hyper-parameter $\gamma$ for the synthesis term in \eqref{our_model}, we employ a bilevel optimization framework which solves for $\Theta$ on training set for any given $\gamma$ in the lower-level problem and tunes the optimal hyper-parameter $\gamma$ on validation set in the upper-level problem. More precisely, our bilevel optimization framework reads as:
	\begin{equation}
		\label{eq:bi-level}
		\begin{aligned}
			\min_{\gamma}  \quad  \textstyle{\sum^{\mathcal{M}_{val}}_{i=1} }\ell( \Theta (\gamma) , \gamma ; \D^{val}_{i}) \quad 
			\mbox{s.t.}  \quad  \Theta(\gamma)  = \argmin_{\Theta} \textstyle{\sum^{\mathcal{M}_{tr}}_{i=1}} \ell ( \Theta , \gamma; \D^{tr}_{i}),
		\end{aligned}
		%\vspace{-10pt}
	\end{equation}
	%where
	\begin{equation}
		\label{eq:def-l}
		\begin{split}
			& \mbox{where}  \quad \ell ( \Theta , \gamma; \D_{i})  :=  \frac{\mu}{2} \|g_{\theta} ([h_{w_1}(\xbf_1^{*i}),h_{ w_2} (\xbf_2^{*i})]) - \xbf_3^{*i}\|_2^2 \\
			& + \textstyle{\sum_{j = 1}^{3}} \Big ( \frac{1}{2} \|\xbf_j^{(\hat{T})}({\Theta , \gamma}; \D_{i})  - \xbf_j^{*i} \|^2_2 
			+ (1 - SSIM(\xbf_j^{(\hat{T})}({\Theta , \gamma}; \D_{i}), \xbf_j^{*i} ) ) \Big ) ,
		\end{split}
	\end{equation}
	and the $\xbf_j^{(\hat{T})}(\cdot)$ denotes the output of the $\hat{T}$-phase network for the $j$th modality.
	The first term of the loss function $\ell$ in \eqref{eq:def-l} presses $g_{\theta}$ to accurately synthesize $\xbf_3$. The second term is to minimize the difference between the network output and the ground truth in the least square sense. The third term is to promote high structural similarity index \cite{wang2004image} of the reconstructed/synthesized images. 
	In \eqref{eq:bi-level}, the lower-level optimization learns the network parameters $\Theta$ with any fixed coefficient $\gamma$ on the training set, and the upper-level tunes the hyper-parameter $\gamma$  on the validation set, which mitigates the challenging overfitting issue. % of existing deep learning approaches where the trained networks may perform poorly on new data sets different from the training set.
	
	The bi-level optimization problem \eqref{eq:bi-level} is very difficult to solve.
	%due to the complex double layer structure of the problem (the constraint is formulated as a separate optimization problem which is by its own difficult to solve) and nonconvexity of objective functions. 
	In this work, we employ the penalty method proposed in \cite{mehra2019penalty} to solve this problem.
	For notation simplicity, we denote $\mathcal{L}(\Theta , \gamma ; \D) := \sum^{\mathcal{M}}_{i=1}\ell (\Theta , \gamma  ; \D_i)$ then rewrite \eqref{eq:bi-level} as
	\begin{equation}
		\min_{ \gamma} \mathcal{L}( \Theta(\gamma), \gamma ; \D^{val}) \ \ \ \ \ \mbox{s.t.} \ \ \Theta(\gamma) =   \argmin_{\Theta}\mathcal{L}( \Theta, \gamma ; \D^{tr}).
		\label{simplified bi-level}
	\end{equation}
	Following \cite{mehra2019penalty}, we relax the lower-level optimization problem to its first-order necessary condition:
	%As the above bi-level optimization problem \eqref{simplified bi-level} is hard to solve, here we first consider to relax it into a single-level constrained optimization problem following \cite{mehra2019penalty}. More specifically, we use the  first-order necessary condition of the lower-level part as a constraint term, as shown below
	%
	\begin{equation}
		\min_{ \gamma} \mathcal{L}( \Theta(\gamma), \gamma ; \D^{val})  \ \ \ \ \ \mbox{s.t.} \ \ \nabla_{\Theta} \mathcal{L}( \Theta(\gamma), \gamma ; \D^{tr}) = 0.
		\label{simplified bi-level-1}
	\end{equation}
	Furthermore, we impose a quadratic penalty on the constraint and further relax the above problem as
	%The problem above can be further relaxed to be unconstrained by replacing the constraint by a penalty term
	\begin{equation}
		\min_{ \Theta, \gamma} \big\{ \widetilde{\mathcal{L}}( \Theta, \gamma ; \D^{tr}, \D^{val}) := \mathcal{L}( \Theta, \gamma ; \D^{val}) + \frac{\lambda}{2} \| \nabla_{\Theta} \mathcal{L}( \Theta, \gamma; \D^{tr}) \|^2 \big\}.
		\label{simplified bi-level-2}
	\end{equation}
	%s\ye{use a notation of the weight different from $\gamma$.}
	Due to the large volume of the datasets, it is not possible to solve \eqref{simplified bi-level-2} in full-batch. Here we train the parameters using the mini-batch stochastic alternating direction method summarized in Algorithm \ref{alg:model}.
	\begin{algorithm}[!htb]
		\caption{Mini-batch alternating direction penalty algorithm}\label{alg:model}
		\begin{algorithmic}[1]
			\STATE \textbf{Input}  $\D^{tr}$, $\D^{val}$, $\delta_{tol}>0$, \textbf{Initialize}  $ \Theta$, $ {\gamma}$, $\delta$, $\lambda>0$ and $\nu_\delta \in(0, 1)$, \ $\nu_\lambda > 1$.
			\WHILE{$\delta > \delta_{tol}$}
			\STATE Sample training and validation batch $\mathcal{B}^{tr} \subset \D^{tr}, \mathcal{B}^{val} \subset \D^{val}$.
			\WHILE{$\|\nabla_{\Theta}\widetilde{\mathcal{L}}( \Theta, \gamma ; \mathcal{B}^{tr}, \mathcal{B}^{val})\|^2 + \| \nabla_{\gamma}\widetilde{\mathcal{L}}( \Theta, \gamma ; \mathcal{B}^{tr}, \mathcal{B}^{val})\| ^2 > \delta$}
			%\STATE \textbf{for} $k=1,2,\dots,K$ (inner loop), \textbf{do} $ \Theta \leftarrow \Theta - \rho_{\Theta}^k \nabla_{\Theta}\widetilde{\mathcal{L}}( \Theta, \gamma ; \mathcal{B}^{tr}, \mathcal{B}^{val})$
			\FOR{$k=1,2,\dots,K$ (inner loop)}
			\STATE $ \Theta \leftarrow \Theta - \rho_{\Theta}^{(k)} \nabla_{\Theta}\widetilde{\mathcal{L}}( \Theta, \gamma ; \mathcal{B}^{tr}, \mathcal{B}^{val})$
			\ENDFOR
			\STATE $ \gamma \leftarrow \gamma - \rho_{\gamma} \nabla_{\gamma}\widetilde{\mathcal{L}}( \Theta, \gamma ; \mathcal{B}^{tr}, \mathcal{B}^{val})$
			\ENDWHILE{ and \textbf{update} $\delta \leftarrow \nu_\delta \delta$, $\ \lambda \leftarrow \nu_\lambda \lambda$. }
			\ENDWHILE{ and \textbf{output:} $\Theta, {\gamma}$.}
			%\STATE 
		\end{algorithmic}
		\label{penelty_method}
	\end{algorithm}
	\begin{figure}[h!]
		\centering 
		\includegraphics[width=.9\textwidth]{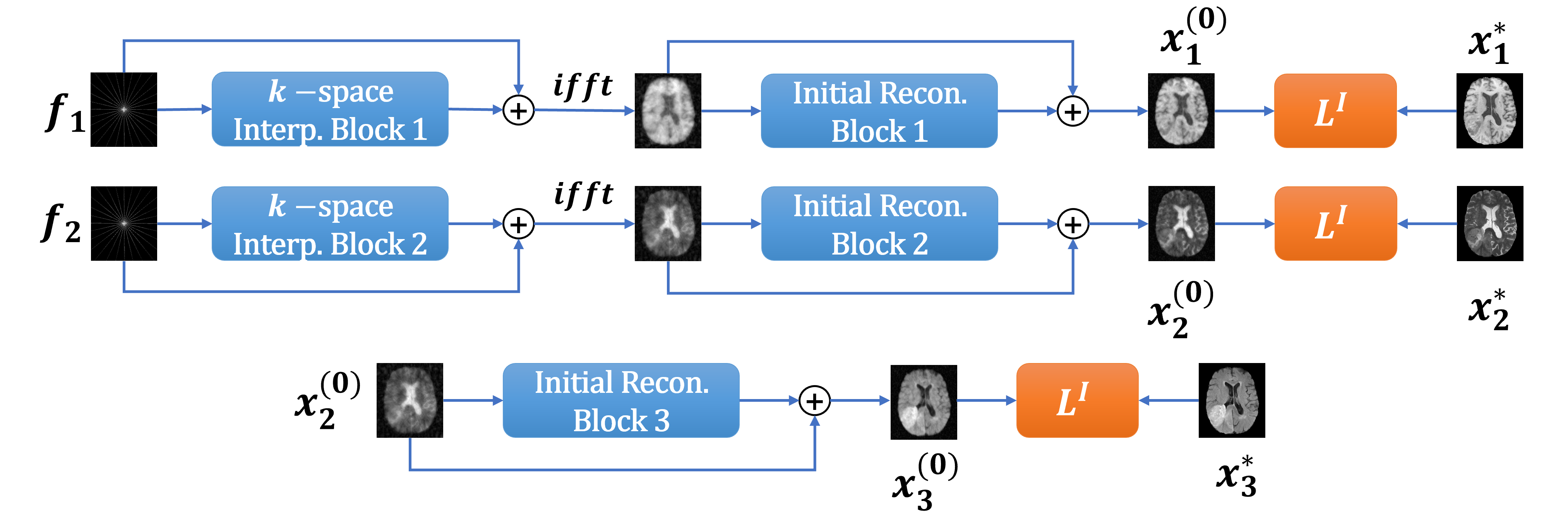}
		\includegraphics[width=.9\textwidth]{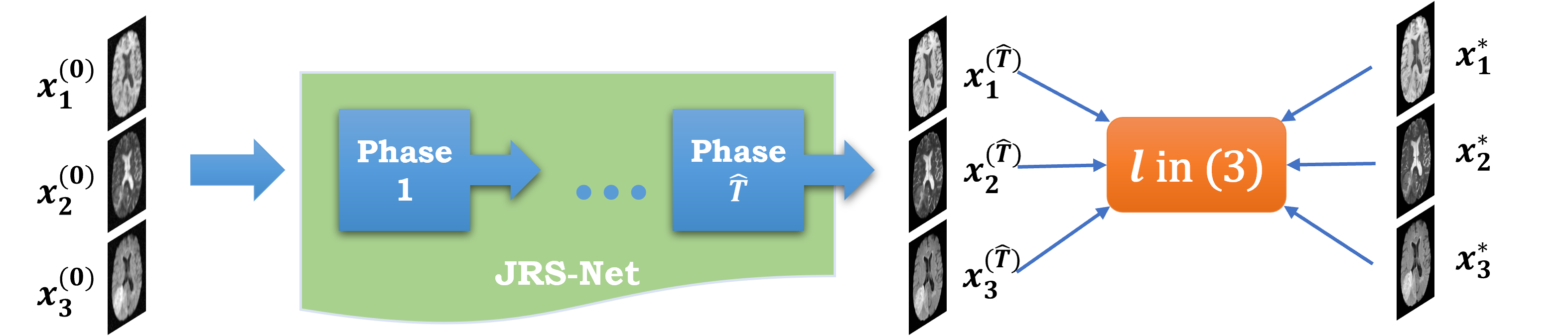}
		\caption{The overall architecture of the proposed network for joint multimodal MRI reconstruction and synthesis: INIT-Nets (up and middle), JRS-Net (bottom).}
		\label{fig:flowchart}
	\end{figure}
	%\vspace{-30pt}
	\section{Experiments}
	\subsection{Initialization Networks}
	\label{Initialization}
	The initials $\{\xbf_1^{(0)}, \xbf_2^{(0)}, \xbf_3^{(0)} \}$ are obtained through the Initialization Networks (INIT-Nets) shown in Fig. \ref{fig:flowchart}, where the \textbf{k-space interpolation block} interpolates the missing components of the undersampled k-space data then fed into the \textbf{initial reconstruction block} in the image domain after inverse Fourier transform. All blocks are designed in residual structure \cite{he2016deep}.  % the multi-phase \textbf{recovery net} which iterates LOA that illustrated in Section \ref{loa induced net}. 
	%Each block of the initialization nets is designed in residual structure. First, we use a 4-convolution \textbf{k-space interpolation block} to interpolate the missing components of the undersampled k-space data $\fbf_1, \fbf_2$. Next we pass the interpolated k-space data into the inverse Fourier transform operator then input into an \textbf{initial reconstruction block} in image domain with 6-convolution separate by ReLU and finally we obtain the refined initial images $\xbf_1^{(0)}$,  $\xbf_2^{(0)}$ of source modalities. 
	%As the k-space data for $\xbf_3^{(0)}$ is missing, we take one of the initial reconstructions $\xbf_1^{(0)}$ and $\xbf_2^{(0)}$ as input to initialize $\xbf_3^{(0)}$.
	%
	To train the INIT-Nets, we minimize the difference between its outputs and the ground truth with loss
	$\label{eq:g}
	L^I(\xbf_j^{(0)}, \xbf_j^{*}) = \|\xbf_j^{(0)} - \xbf_j^{*} \|_1, \ j = 1, 2, 3.
	$
	The INIT-Nets only produce initial approximate images with limited accuracy, so we fed them into the Joint Reconstruction and Synthesis Network (JRS-Net) illustrated in Section \ref{loa induced net} to obtain the final results. INIT-Nets are pre-trained whose parameters are fixed during training the JRS-Net.
	\begin{table}[htb]
		\centering
		\caption{Quantitative comparison of the synthesis results. }\label{quant_results}
		%\addtolength{\tabcolsep}{-1pt}
		\begin{tabular}{l|l|c|c|c}
			\hline
			& \textbf{Methods} &  \textbf{PSNR} & \textbf{SSIM} & \textbf{NMSE}\\
			%\midrule
			\hline
			T1 $+$ T2 $\to$ FLAIR  & MM \cite{chartsias2017multimodal}  & 22.89$\pm$1.48 &  0.6671$\pm$0.0586 &  0.0693$\pm$0.0494 \\
			& MM-GAN \cite{sharma2019missing}  & 23.35$\pm$1.03 &  0.7084$\pm$0.0370 &  0.0620$\pm$0.0426 \\
			& MMGradAdv \cite{liu2020multimodal}  & 24.03$\pm$1.40 & 0.7586$\pm$0.0326 & 0.0583$\pm$0.0380 \\
			& Hi-Net \cite{zhou2020hi} &   25.03$\pm$1.38 & 0.8499$\pm$0.0300 & 0.0254$\pm$0.0097\\
			& Proposed &  \textbf{26.19$\pm$1.34} & \textbf{0.8677$\pm$0.0307} & \textbf{0.0205$\pm$0.0087}\\
			\hline
			$\fbf_{T1} + \fbf_{T2} \to$ FLAIR & Proposed & \textbf{25.74$\pm$1.25} & \textbf{0.8597$\pm$0.0315} & \textbf{0.0215$\pm$0.0085}\\
			\hline
			T1 $+$ FLAIR  $\to$ T2 & MM \cite{chartsias2017multimodal}  & 23.89$\pm$1.61 & 0.6895 $\pm$0.0511 & 0.0494$\pm$0.0185 \\
			& MM-GAN \cite{sharma2019missing}  & 24.15$\pm$0.90 & 0.7217$\pm$0.0432 & 0.0431$\pm$0.0114 \\
			& MMGradAdv \cite{liu2020multimodal} & 25.06$\pm$1.49 &  0.7597$\pm$0.0486 & 0.0406$\pm$ 0.0165\\ 
			& Hi-Net \cite{zhou2020hi} & 25.95$\pm$1.50 & 0.8552$\pm$0.0410 & 0.0229$\pm$0.0070\\
			\hline
			$\fbf_{T1} + \fbf_{FLAIR} \to$ T2 & Proposed & \textbf{26.52$\pm$1.57} & \textbf{0.8610$\pm$0.0438} & \textbf{0.0207$\pm$0.0072}\\
			\hline
			T2 $+$ FLAIR  $\to$ T1 & MM \cite{chartsias2017multimodal} &  23.53$\pm$2.18 & 0.7825$\pm$0.0470 &  0.0301$\pm$0.0149 \\
			& MM-GAN \cite{sharma2019missing}  & 23.63$\pm$2.31  & 0.7908$\pm$0.0421 & 0.0293 $\pm$0.0119\\
			& MMGradAdv \cite{liu2020multimodal} & 24.73$\pm$2.23 & 0.8065$\pm$0.0423 &  0.0252$\pm$0.0118 \\ 
			& Hi-Net \cite{zhou2020hi} &  25.64$\pm$1.59 & 0.8729$\pm$0.0349 &  0.0130$\pm$0.0097 \\
			\hline
			$\fbf_{T2} + \fbf_{FLAIR} \to$ T1 & Proposed &  \textbf{26.31$\pm$1.80} & \textbf{0.9085$\pm$0.0311} & \textbf{0.0112$\pm$0.0113} \\
			\hline
			T1 $+$ T2 $\to$ T1CE & MM \cite{chartsias2017multimodal} &  23.37$\pm$1.56 & 0.7272$\pm$0.0574 &  0.0312$\pm$ 0.0138 \\
			& MM-GAN \cite{sharma2019missing}  & 23.68$\pm$0.97 & 0.7577$\pm$0.0637 & 0.0302$\pm$0.0133\\
			& MMGradAdv \cite{liu2020multimodal} &  24.23$\pm$1.90 & 0.7887$\pm$0.0519  & 0.0273$\pm$0.0136   \\ 
			& Hi-Net \cite{zhou2020hi} &  25.21$\pm$1.20 & 0.8650$\pm$0.0328  &  0.0180$\pm$0.0134   \\
			\hline
			$\fbf_{T1} + \fbf_{T2} \to$ T1CE & Proposed &  \textbf{25.91$\pm$1.21} & \textbf{0.8726$\pm$0.0340}  &  \textbf{0.0167$\pm$0.0133}   \\
			\hline
		\end{tabular}
	\end{table}
	\begin{table}[!h]
		%\vspace{10pt}
		\centering
		\caption{Quantitative comparison of the reconstructed T1 and T2 images without and with joint synthesis of FLAIR images.} \label{t1t2}
		\addtolength{\tabcolsep}{5pt}
		\begin{tabular}{c|c|c|c|c}
			\hline
			\textbf{Modality} & \textbf{FLAIR involved?} & \textbf{PSNR} & \textbf{SSIM} & \textbf{NMSE}\\
			\hline
			T1 & No &  37.00$\pm$0.74 & 0.9605$\pm$0.0047 & 0.0008$\pm$0.0002 \\
			& Yes & 37.49$\pm$0.83 &  0.9628$\pm$0.0074 & 0.0007$\pm$0.0002 \\
			T2 & No &  37.24$\pm$1.22 & 0.9678$\pm$0.0028 & 0.0027$\pm$0.0010 \\
			& Yes & 37.67$\pm$1.34 & 0.9663$\pm$0.0043 & 0.0024$\pm$0.0009 \\ 
			\hline
		\end{tabular}
	\end{table}
	
	\subsection{Experiment Setup}
	The datasets are from BRATS 2018 challenge \cite{menze2014multimodal} which were scanned from four modalities T1, T2, Flair and T1-weighted contrast-enhanced (T1CE) and we picked high-grade gliomas (HGG) set which consists 210 patients. We applied Fourier transform to the images and undersampled the k-space using a radial mask of sampling ratio $40 \%$ to obtain partial k-space data for training.  We randomly took the center 10 slices from 6 patients as test set with cropped size $ 160 \times 180$ and split the rest of HGG dataset into training and validation sets with 1020 images separately. We compared with four state-of-the-art multimodal MR synthesis methods: Multimodal MR (MM) \cite{chartsias2017multimodal}, MM-GAN \cite{sharma2019missing}, MMGradAdv \cite{liu2020multimodal} and Hi-Net \cite{zhou2020hi}. 
	%The inputs for our method are the partial k-sapce data scanned with a radial mask with sampling ratio $40 \%$. 
	%
	The hyper-parameter selection for our algorithm is provided in Supplementary Material. 
	Three metrics are used for evaluation: {peak signal-to-noise ratio (PSNR) \cite{hore2010image}}, structural similarity (SSIM) \cite{wang2004image}, and normalized mean squared error {(}NMSE{) \cite{NMSE}}.
	%We implemented their methods with the hyper-parameters indicated in their papers and evaluated all methods using the same training and test sets. The input images of the above four methods are reference images of source modalities, while t We normalized the intensity value of each image into a range of $[-1,1]$. 
	\begin{figure}[h!]
		\centering
		\includegraphics[width=0.03\linewidth]{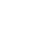}
		\includegraphics[width=0.15\linewidth]{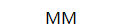}
		\includegraphics[width=0.15\linewidth]{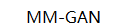}
		\includegraphics[width=0.15\linewidth]{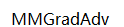}
		\includegraphics[width=0.15\linewidth]{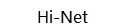}
		\includegraphics[width=0.15\linewidth]{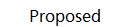}
		\includegraphics[width=0.15\linewidth]{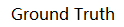}\\
		\includegraphics[width=0.03\linewidth]{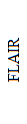}
		\includegraphics[width=0.15\linewidth]{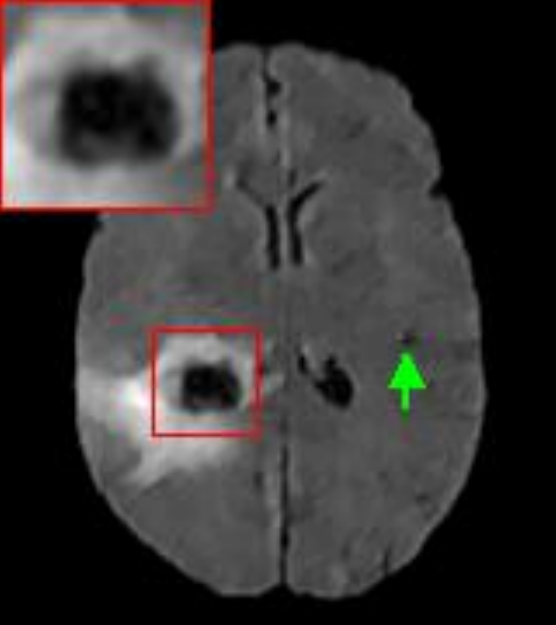}
		\includegraphics[width=0.15\linewidth]{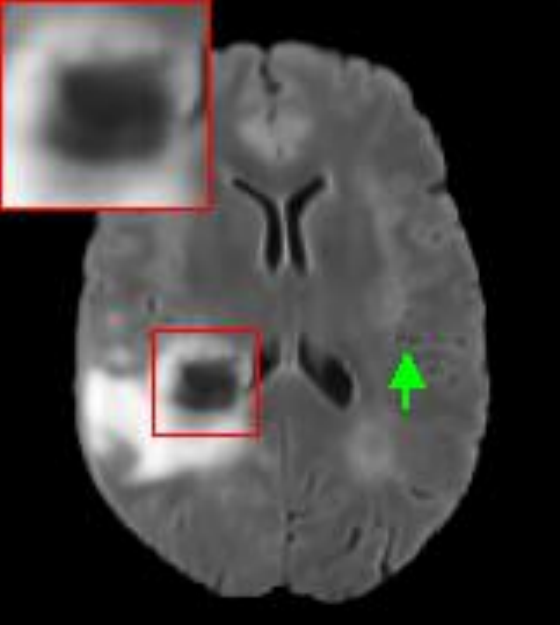}
		\includegraphics[width=0.15\linewidth]{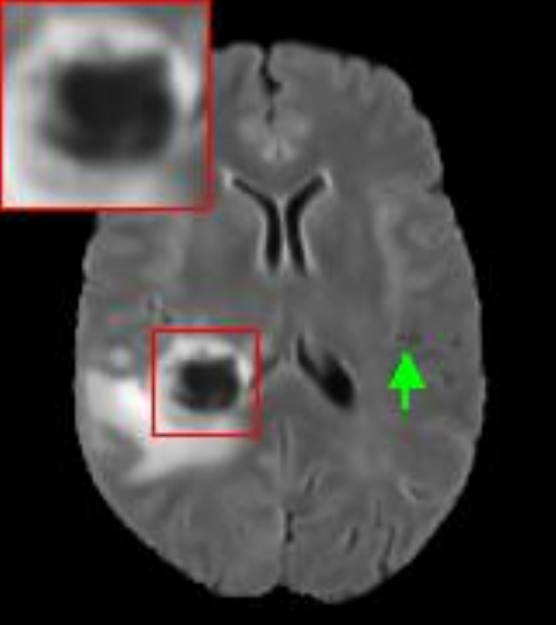}
		\includegraphics[width=0.15\linewidth]{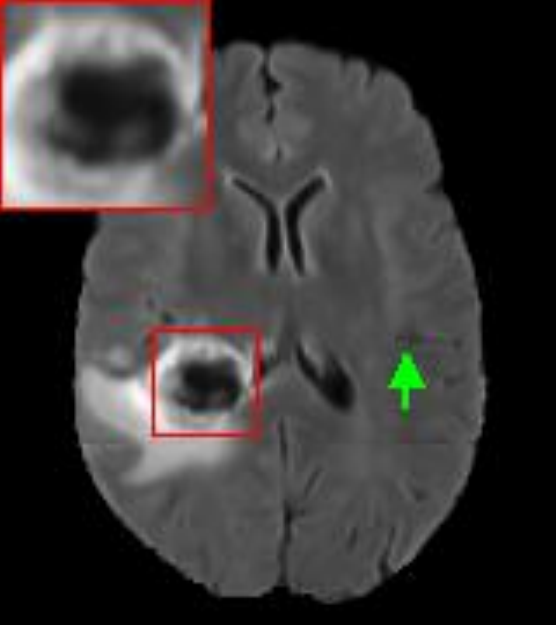}
		\includegraphics[width=0.15\linewidth]{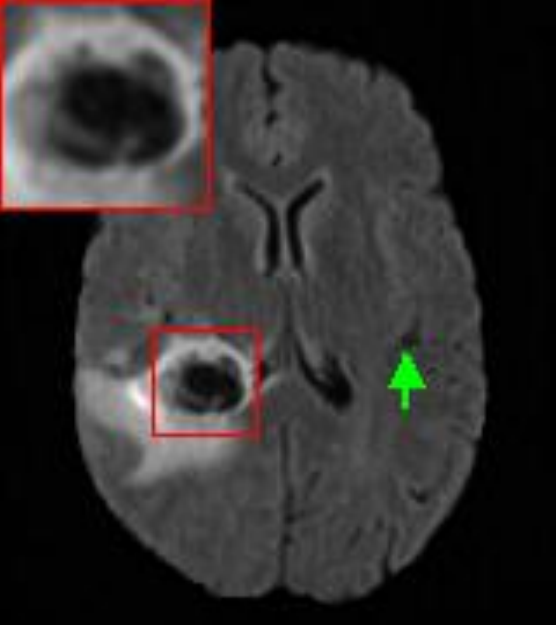}
		\includegraphics[width=0.15\linewidth]{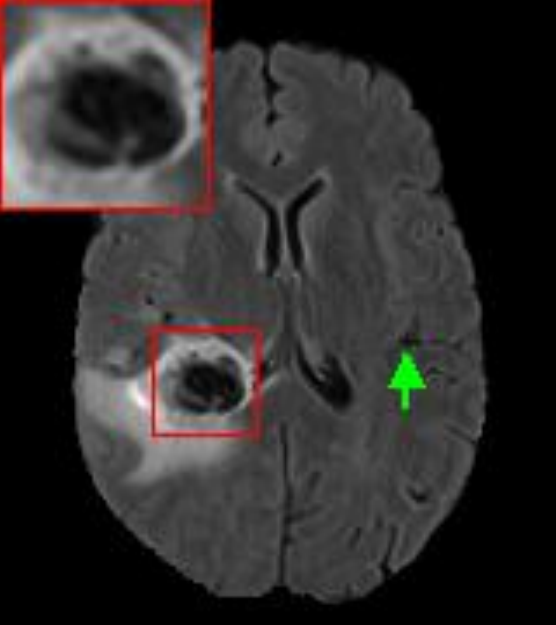}\\
		\includegraphics[width=0.03\linewidth]{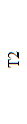}
		\includegraphics[width=0.15\linewidth]{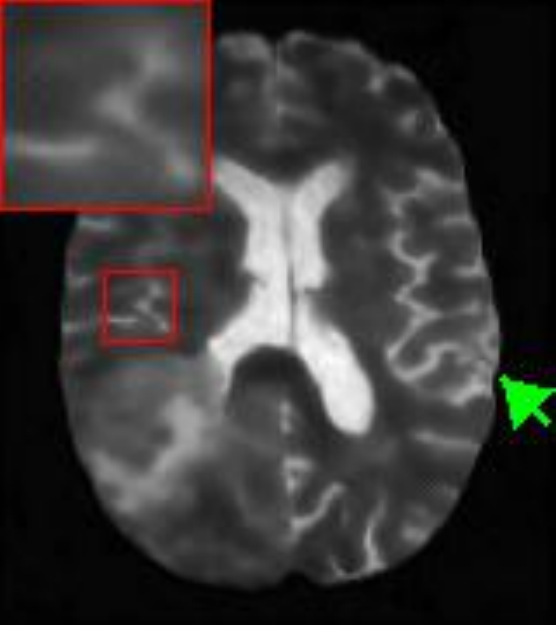}
		\includegraphics[width=0.15\linewidth]{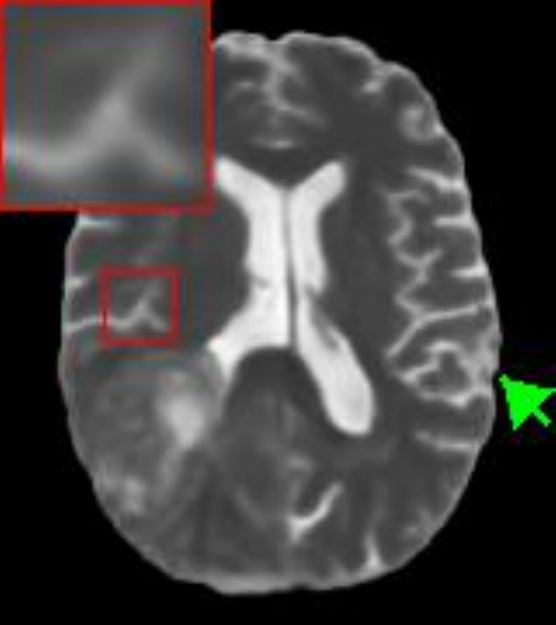}
		\includegraphics[width=0.15\linewidth]{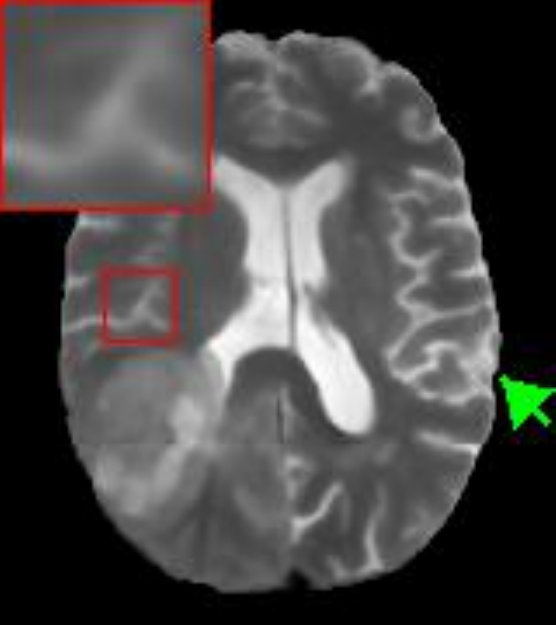}
		\includegraphics[width=0.15\linewidth]{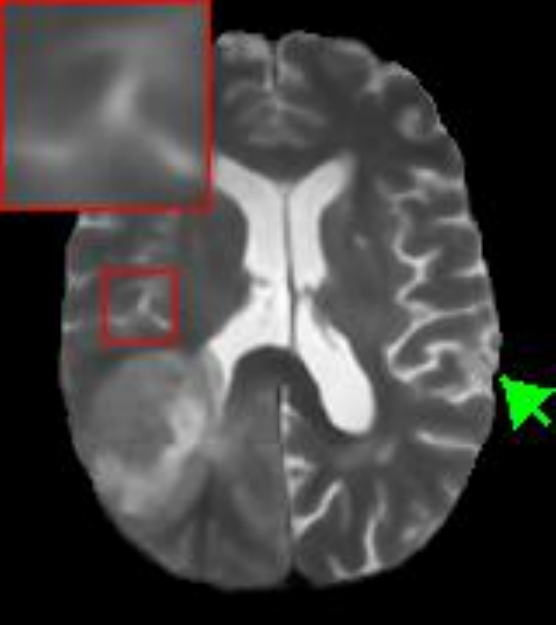}
		\includegraphics[width=0.15\linewidth]{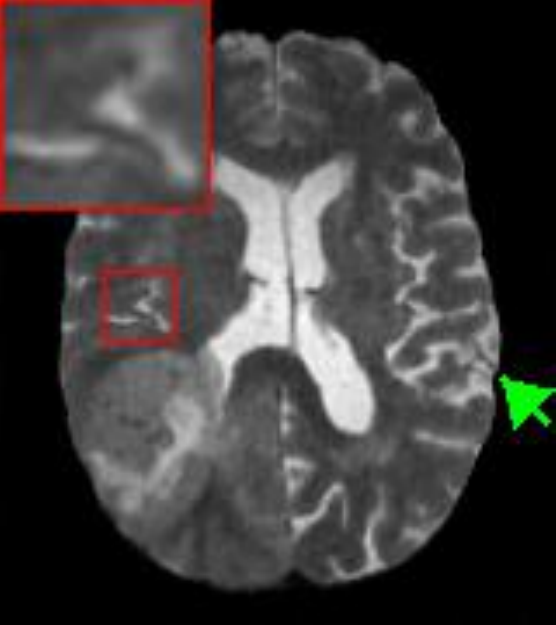}
		\includegraphics[width=0.15\linewidth]{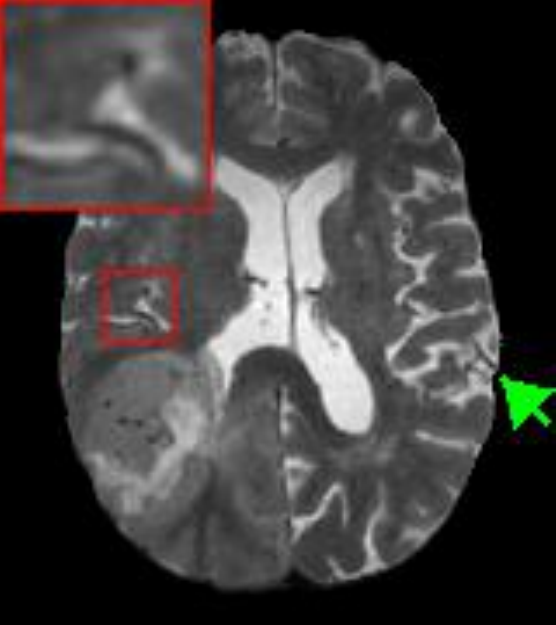}\\
		\includegraphics[width=0.03\linewidth]{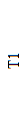}
		\includegraphics[width=0.15\linewidth]{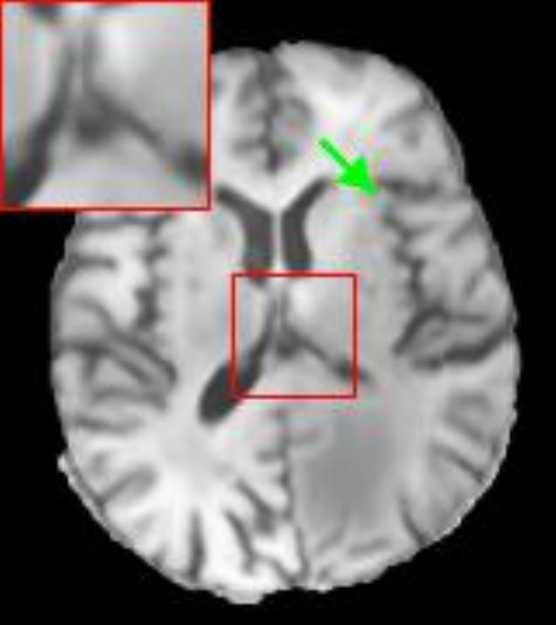}
		\includegraphics[width=0.15\linewidth]{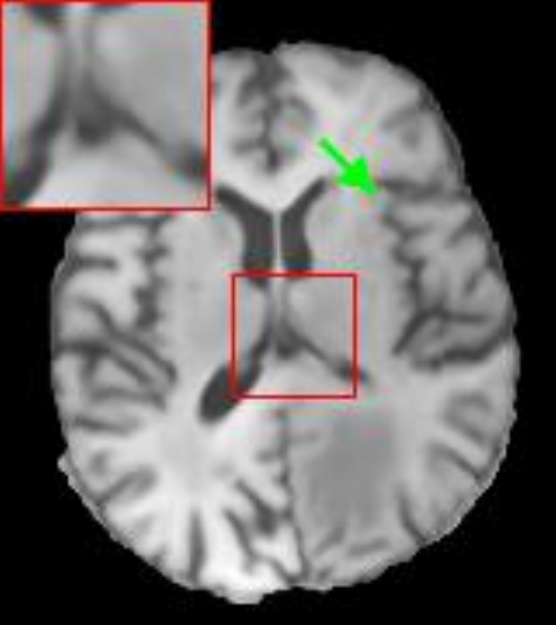}
		\includegraphics[width=0.15\linewidth]{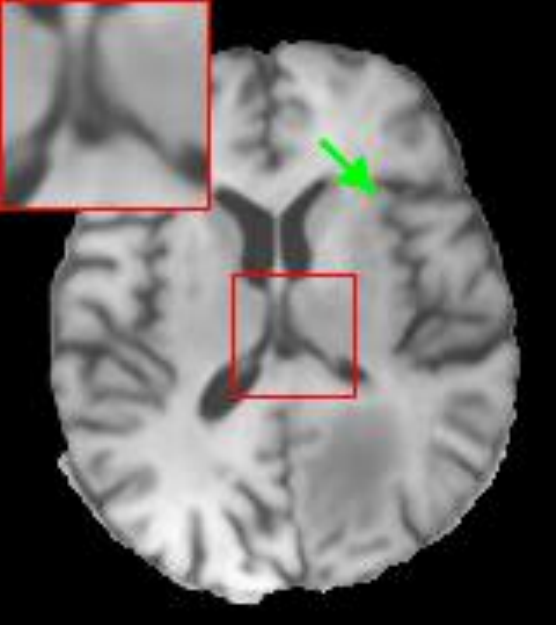}
		\includegraphics[width=0.15\linewidth]{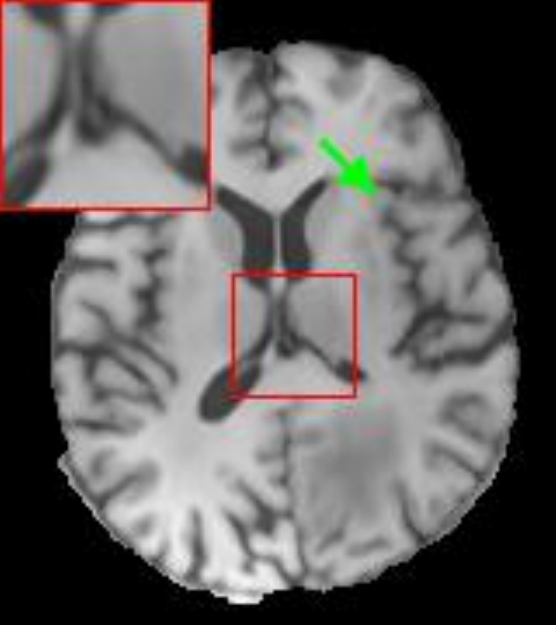}
		\includegraphics[width=0.15\linewidth]{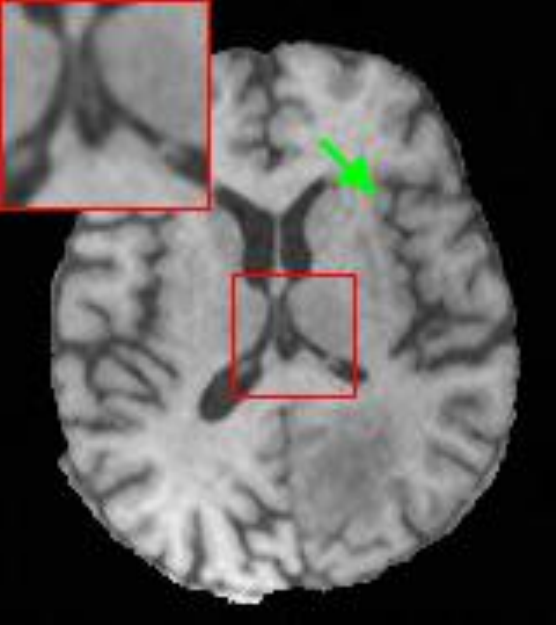}
		\includegraphics[width=0.15\linewidth]{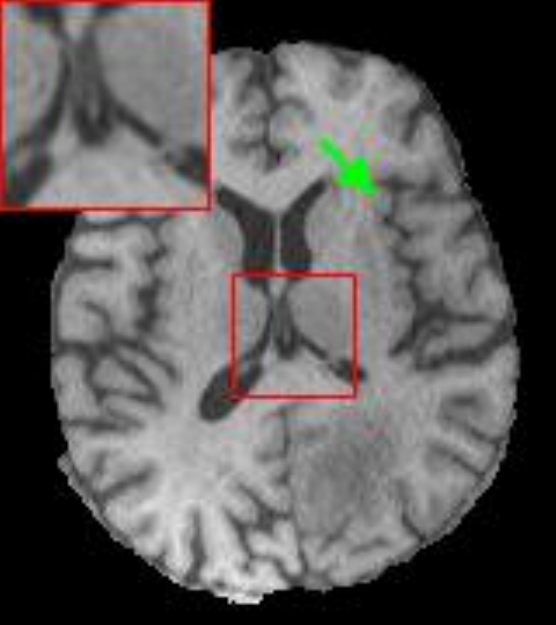}\\
		\includegraphics[width=0.03\linewidth]{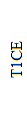}
		\includegraphics[width=0.15\linewidth]{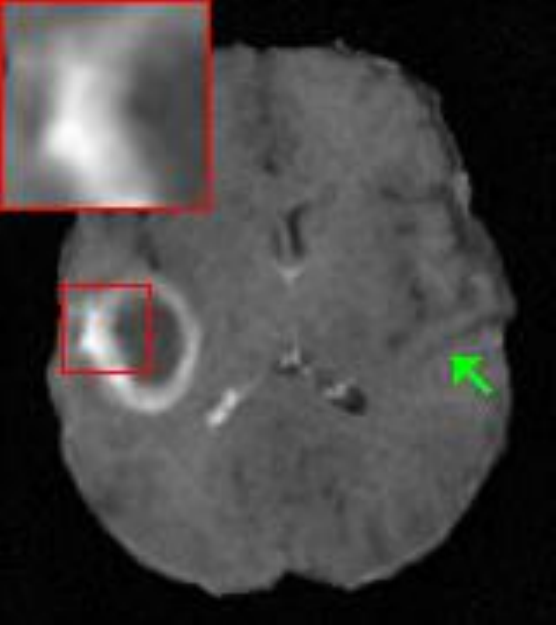}
		\includegraphics[width=0.15\linewidth]{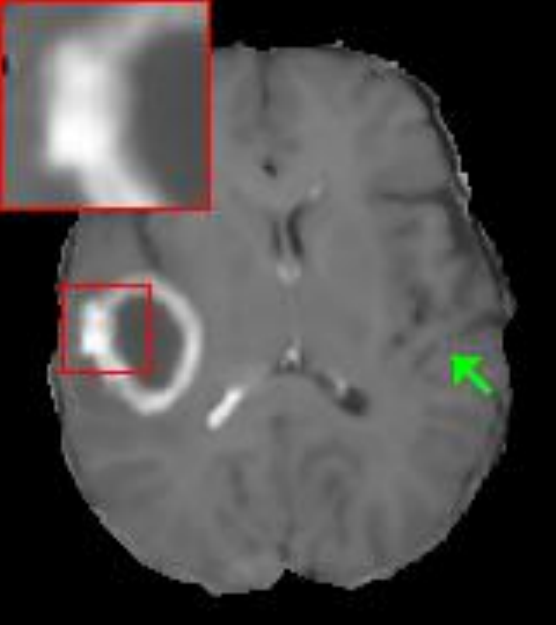}
		\includegraphics[width=0.15\linewidth]{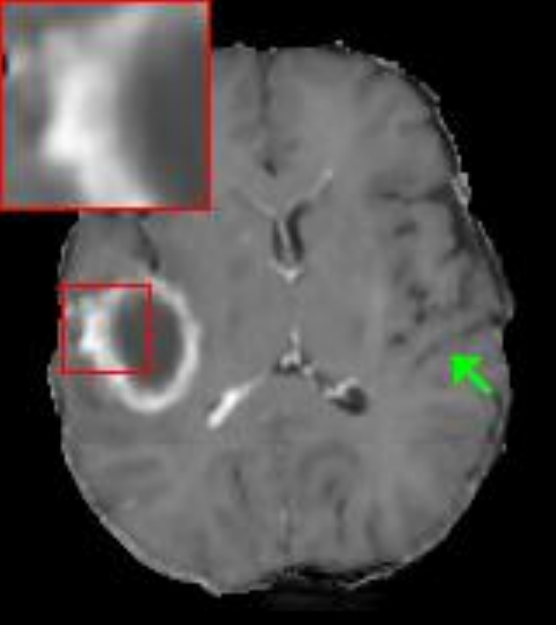}
		\includegraphics[width=0.15\linewidth]{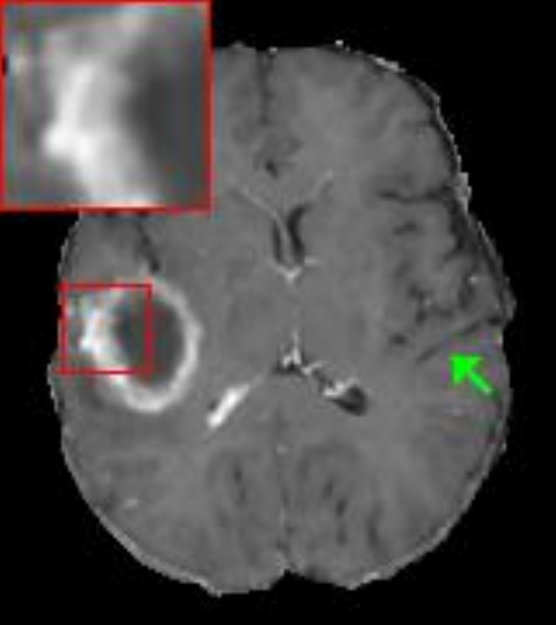}
		\includegraphics[width=0.15\linewidth]{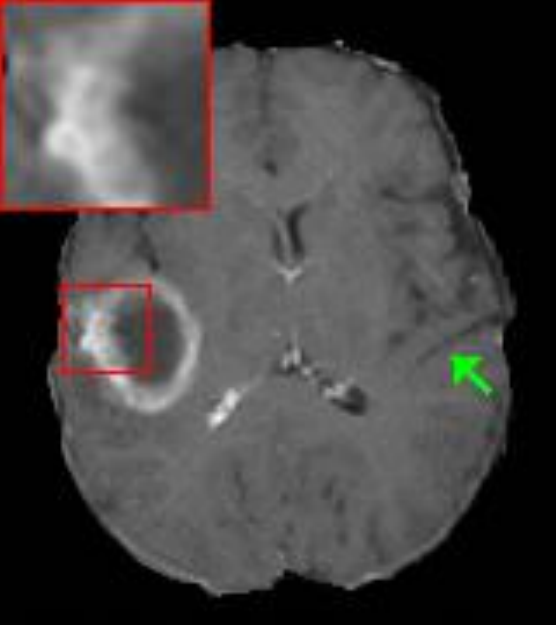}
		\includegraphics[width=0.15\linewidth]{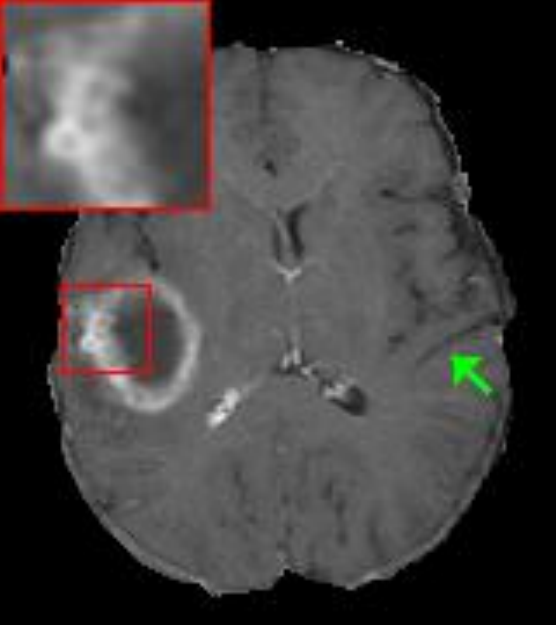}\\
		\caption{Qualitative comparison between the state-of-the-art multimodal synthesis methods and proposed method. From first row to last row: T1 $+$ T2 $\to$ FLAIR, T1 $+$ FLAIR $\to$ T2, T2 $+$ FLAIR $\to$ T1 and T1 $+$ T2 $\to$ T1CE.  }\label{fig:synthesis}
	\end{figure}

	\subsection{Experimental Results and Evaluation}
	We take four synthesis directions T1 $+$ T2 $\to$ FLAIR,  T1 $+$ FLAIR  $\to$ T2,  T2 $+$ FLAIR  $\to$ T1 and T1 $+$ T2 $\to$ T1CE.  Table \ref{quant_results} reports the quantitative result, which indicates that the proposed method outperforms MM and GAN-based methods (MM-GAN, MMGradAdv, Hi-Net). The average PSNR of our method improves 0.67 dB comparing to the baseline Hi-Net, SSIM improves about 0.01,  and NMSE reduces about 0.01.  We also conduct the pure-synthesis experiment (T1 $+$ T2 $\to$ FLAIR) by inputting fully-scanned source data, where the model \eqref{our_model} minimizes w.r.t. $\xbf_3$ only and excludes the data-fidelity terms. This result is in Table \ref{quant_results} where the PSNR value is 1.16 dB higher than baseline method Hi-Net. %In addition, the synthesized target modality is generated from the partial k-space data of two different source modalities while other methods are learned from the true data. Our model first learn a good initial input of the Algorithm \ref{alg:lda} and then get updated reconstruction images and synthesis images by iterating the algorithm, which has theoretical convergence guarantee and thus achieves better performance. 
	%Table \ref{t1t2} compares purely reconstructing source modalities (only minimize w.r.t. $\xbf_1$ and $\xbf_2$ in model \eqref{our_model}) and jointly reconstructing source modalities  and synthesizing target modality (minimize w.r.t. all $\xbf_1$, $\xbf_2$ and $\xbf_3$).
	%The results are listed in Table \ref{t1t2}, which shows the comparison between purely reconstruct source modalities (only minimize $\xbf_1$ and $\xbf_2$ in model \eqref{our_model}) and jointly reconstruct source modalities  and synthesis target modality (minimize both $\xbf_1$, $\xbf_2$ and $\xbf_3$). 
	Table \ref{t1t2} shows that the joint reconstruction and synthesis improves PSNR by 0.46 dB comparing to purely reconstructing T1 and T2 without synthesizing FLAIR. 
	%\textcolor{blue}{The reason is due to the third term in \eqref{our_model} has controlled by a weighted parameter $\gamma$, which is trained by bilevel optimization Algorithm \ref{alg:model} to promote the performance of all the three modalities.
	%Also $g_\theta$ and $h_{w_3}$ generate more features and useful information of $\xbf_3$ which benefits the reconstruction of $\xbf_1$ and $\xbf_2$. This demonstrates the superiority of the proposed method in both MRI reconstruction and synthesis and these two tasks are mutually beneficial.}
	%\textcolor{red}{Synthesis using common features and learn the mapping use data. x1 and x2 the joint feature maps can be refined by the information from x3 if g theta is good. }
	We think this is because that the synthesis operator $g_{\theta}$ also leverages data $\xbf_3$ to assist shaping the feature maps of $\xbf_1$ and $\xbf_2$, which improves the reconstruction quality of the latter images.
	
	Fig. \ref{fig:synthesis} displays the synthetic MRI results on different source and target modality images. The proposed synthetic images preserve more details and distinct edges of the tissue boundary (indicated by the magnified red windows and green arrows) and the synthetic images are more alike the ground truth images comparing to other referenced methods. 
	%
	%\noindent
	\section{Conclusion}
	We propose a novel deep model that simultaneously reconstructs the source modality images from the partially scanned k-space MR data and synthesizes the target modality image without any k-space information by iterating an LOA with convergence guaranteed. The network is trained by a bilevel-optimization training algorithm that uses training and validation sets to further improve the performance. 
	Extensive experiments on brain MR data with different modalities validate the magnificent performance of the proposed model.
	
	\bibliographystyle{splncs04}
	\bibliography{paper63}

\begin{thebibliography}{10}
\providecommand{\url}[1]{\texttt{#1}}
\providecommand{\urlprefix}{URL }
\providecommand{\doi}[1]{https://doi.org/#1}

\bibitem{pmlr-v119-bao20b}
Bao, R., Gu, B., Huang, H.: Fast {OSCAR} and {OWL} regression via safe
  screening rules. In: III, H.D., Singh, A. (eds.) Proceedings of the 37th
  International Conference on Machine Learning. Proceedings of Machine Learning
  Research, vol.~119, pp. 653--663. PMLR (13--18 Jul 2020)

\bibitem{bian2022optimal}
Bian, W., Chen, Y., Ye, X.: An optimal control framework for joint-channel
  parallel mri reconstruction without coil sensitivities. Magnetic Resonance
  Imaging  (2022)

\bibitem{jimaging7110231}
Bian, W., Chen, Y., Ye, X., Zhang, Q.: An optimization-based meta-learning
  model for mri reconstruction with diverse dataset. Journal of Imaging
  \textbf{7}(11) (2021)

\bibitem{chartsias2017multimodal}
Chartsias, A., Joyce, T., Giuffrida, M.V., Tsaftaris, S.A.: Multimodal mr
  synthesis via modality-invariant latent representation. IEEE transactions on
  medical imaging  \textbf{37}(3),  803--814 (2017)

\bibitem{chen2021learnable}
Chen, Y., Liu, H., Ye, X., Zhang, Q.: Learnable descent algorithm for nonsmooth
  nonconvex image reconstruction. SIAM Journal on Imaging Sciences
  \textbf{14}(4),  1532--1564 (2021)

\bibitem{cheng2019model}
Cheng, J., Wang, H., Ying, L., Liang, D.: Model learning: Primal dual networks
  for fast mr imaging. In: International Conference on Medical Image Computing
  and Computer-Assisted Intervention. pp. 21--29. Springer (2019)

\bibitem{dar2019image}
Dar, S.U., Yurt, M., Karacan, L., Erdem, A., Erdem, E., {\c{C}}ukur, T.: Image
  synthesis in multi-contrast mri with conditional generative adversarial
  networks. IEEE transactions on medical imaging  \textbf{38}(10),  2375--2388
  (2019)

\bibitem{dar2020prior}
Dar, S.U., Yurt, M., Shahdloo, M., Ild{\i}z, M.E., T{\i}naz, B., {\c{C}}ukur,
  T.: Prior-guided image reconstruction for accelerated multi-contrast mri via
  generative adversarial networks. IEEE Journal of Selected Topics in Signal
  Processing  \textbf{14}(6),  1072--1087 (2020)

\bibitem{glorot2010understanding}
Glorot, X., Bengio, Y.: Understanding the difficulty of training deep
  feedforward neural networks. In: Proceedings of the thirteenth international
  conference on artificial intelligence and statistics. pp. 249--256. JMLR
  Workshop and Conference Proceedings (2010)

\bibitem{he2016deep}
He, K., Zhang, X., Ren, S., Sun, J.: Deep residual learning for image
  recognition. In: Proceedings of the IEEE conference on computer vision and
  pattern recognition. pp. 770--778 (2016)

\bibitem{hore2010image}
Hore, A., Ziou, D.: Image quality metrics: Psnr vs. ssim. In: 2010 20th
  international conference on pattern recognition. pp. 2366--2369. IEEE (2010)

\bibitem{iglesias2013synthesizing}
Iglesias, J.E., Konukoglu, E., Zikic, D., Glocker, B., Van~Leemput, K., Fischl,
  B.: Is synthesizing mri contrast useful for inter-modality analysis? In:
  International Conference on Medical Image Computing and Computer-Assisted
  Intervention. pp. 631--638. Springer (2013)

\bibitem{kingma2014adam}
Kingma, D.P., Ba, J.: Adam: {A} method for stochastic optimization. In: Bengio,
  Y., LeCun, Y. (eds.) 3rd International Conference on Learning
  Representations, {ICLR} 2015, San Diego, CA, USA, May 7-9, 2015, Conference
  Track Proceedings (2015)

\bibitem{lee2018deep}
Lee, D., Yoo, J., Tak, S., Ye, J.C.: Deep residual learning for accelerated mri
  using magnitude and phase networks. IEEE Transactions on Biomedical
  Engineering  \textbf{65}(9),  1985--1995 (2018)

\bibitem{liu2020multimodal}
Liu, X., Yu, A., Wei, X., Pan, Z., Tang, J.: Multimodal mr image synthesis
  using gradient prior and adversarial learning. IEEE Journal of Selected
  Topics in Signal Processing  \textbf{14}(6),  1176--1188 (2020)

\bibitem{mehra2019penalty}
Mehra, A., Hamm, J.: Penalty method for inversion-free deep bilevel
  optimization. arXiv preprint arXiv:1911.03432  (2019)

\bibitem{menze2014multimodal}
Menze, B.H., Jakab, A., Bauer, S., Kalpathy-Cramer, J., Farahani, K., Kirby,
  J., Burren, Y., Porz, N., Slotboom, J., Wiest, R., et~al.: The multimodal
  brain tumor image segmentation benchmark (brats). IEEE transactions on
  medical imaging  \textbf{34}(10),  1993--2024 (2014)

\bibitem{monga2021algorithm}
Monga, V., Li, Y., Eldar, Y.C.: Algorithm unrolling: Interpretable, efficient
  deep learning for signal and image processing. IEEE Signal Processing
  Magazine  \textbf{38}(2),  18--44 (2021)

\bibitem{NMSE}
Poli, A., Cirillo, M.: On the use of the normalized mean square error in
  evaluating dispersion model performance. Atmospheric Environment. Part A.
  General Topics  \textbf{27},  2427--2434 (10 1993)

\bibitem{sharma2019missing}
Sharma, A., Hamarneh, G.: Missing mri pulse sequence synthesis using
  multi-modal generative adversarial network. IEEE transactions on medical
  imaging  \textbf{39}(4),  1170--1183 (2019)

\bibitem{sohail2019unpaired}
Sohail, M., Riaz, M.N., Wu, J., Long, C., Li, S.: Unpaired multi-contrast mr
  image synthesis using generative adversarial networks. In: International
  Workshop on Simulation and Synthesis in Medical Imaging. pp. 22--31. Springer
  (2019)

\bibitem{sriram2020end}
Sriram, A., Zbontar, J., Murrell, T., Defazio, A., Zitnick, C.L., Yakubova, N.,
  Knoll, F., Johnson, P.: End-to-end variational networks for accelerated mri
  reconstruction. In: International Conference on Medical Image Computing and
  Computer-Assisted Intervention. pp. 64--73. Springer (2020)

\bibitem{sun2016deep}
Sun, J., Li, H., Xu, Z., et~al.: Deep admm-net for compressive sensing mri.
  Advances in neural information processing systems  \textbf{29} (2016)

\bibitem{WANG2020136}
Wang, S., et~al.: Deepcomplexmri: Exploiting deep residual network for fast
  parallel mr imaging with complex convolution. Magnetic Resonance Imaging
  \textbf{68},  136 -- 147 (2020)

\bibitem{wang2004image}
Wang, Z., et~al.: Image quality assessment: from error visibility to structural
  similarity. IEEE transactions on image processing  \textbf{13}(4),  600--612
  (2004)

\bibitem{welander2018generative}
Welander, P., Karlsson, S., Eklund, A.: Generative adversarial networks for
  image-to-image translation on multi-contrast mr images-a comparison of
  cyclegan and unit. arXiv preprint arXiv:1806.07777  (2018)

\bibitem{yang2017dagan}
Yang, G., Yu, S., Dong, H., Slabaugh, G., Dragotti, P.L., Ye, X., Liu, F.,
  Arridge, S., Keegan, J., Guo, Y., et~al.: Dagan: Deep de-aliasing generative
  adversarial networks for fast compressed sensing mri reconstruction. IEEE
  transactions on medical imaging  \textbf{37}(6),  1310--1321 (2017)

\bibitem{yang2020model}
Yang, Y., Wang, N., Yang, H., Sun, J., Xu, Z.: Model-driven deep attention
  network for ultra-fast compressive sensing mri guided by cross-contrast mr
  image. In: International Conference on Medical Image Computing and
  Computer-Assisted Intervention. pp. 188--198. Springer (2020)

\bibitem{zhang2018ista}
Zhang, J., Ghanem, B.: Ista-net: Interpretable optimization-inspired deep
  network for image compressive sensing. In: Proceedings of the IEEE conference
  on computer vision and pattern recognition. pp. 1828--1837 (2018)

\bibitem{zhou2020hi}
Zhou, T., Fu, H., Chen, G., Shen, J., Shao, L.: Hi-net: hybrid-fusion network
  for multi-modal mr image synthesis. IEEE transactions on medical imaging
  \textbf{39}(9),  2772--2781 (2020)

\end{thebibliography}

	\section{Supplementary}
	%\subsection{The illustration of $h_{w_i}$ and $g_{\theta}$}
	\section{Hyper-parameter Selection}
	All experiments are implemented on Windows workstation with Nvidia GTX-1080Ti GPUs and the parameters are initialized with Xavier initialization \cite{glorot2010understanding} and trained with ADAM optimizer \cite{kingma2014adam} with initial learning rate $0.001$. We put $\gamma = 1$ in model (1), $\mu = 0.1$ in loss $\ell$. In our experiment, we use all complex convolution operators \cite{WANG2020136} where we set $l=4$ convolutions with kernel size $3 \times 3 \times 64$ in $h_{w_i}$ and $l' = 6$ convolutions with kernel size $3 \times 3 \times 128$ in $g_{\theta}$. For Algorithm \ref{alg:lda}, considering both algorithm convergence and the computational efficiency, we take the parameters as follows after trials: $\alpha_0= 0.01, \eta_0=0.01, \varepsilon_0 = 0.001, a = 10^5, \sigma = 10^3, \rho =0.9$. We also set the termination tolerance $\etol = 1\times 10^{-3}$, together with the termination condition defined in Line 5, the algorithm stops at 11 phases once $\etol$ satisfies the stopping criteria. 
	Similarly, in Algorithm \ref{alg:model}, we select the parameters as follows: $\nu_\delta =0.95 $, $\delta =1 \times 10^{-3}$, $\lambda =10^{-4}$, $\nu_\lambda = 1.001$ and $ \rho_\gamma = 0.9$.  We decide the batch size to be $2$ considering the GPU memory and data size. We set $\delta_{tol} = 4.35 \times 10 ^ {-6}$ which makes the algorithm stop at around 1000 epochs. 
	%We apply 4 convolutions to k-space interpolation block and 6-convolution separate by ReLU in initial reconstruction block in image domain. We will provide the publicly shared code depending on acceptance.
	
	\subsection{Proof of Theorem 1}
	We assume that $\Psi_{\Theta, \gamma}$ is coercive and $\Psi_{\Theta, \gamma}^* = \min_{\Xbf} \Psi_{\Theta, \gamma}(\Xbf) > -\infty$, which are easy to satisfy in practice. For any set $\Scal \subset \mathbb{R}^n$, we denote $\mathrm{dist}(\ybf, \Scal) := \inf\{ \|\ybf - \xbf\| \ \vert \ \xbf \in \Scal \}$. The definition of the Clark subdifferential and Clark stationary point can be found in [4]. We need the following lemmas to prove Theorem 1. Since Lemma \ref{lem:bound} can be verified by direct calculation, and Lemma \ref{lem:inner} and \ref{lem:phi_decay} are similar as those given in \cite{chen2021learnable,jimaging7110231}, we omit their proofs here.
	\begin{lemma}\label{r_lips}
		The gradient of $\Psi_{\Theta, \gamma}^{\varepsilon}(\Xbf)$ is Lipschitz continuous.
	\end{lemma}
	\begin{proof}
		Notice that $\Psi_{\Theta, \gamma}^{\varepsilon}(\Xbf)$ is the smoothing surrogate of the $\Psi_{\Theta, \gamma}(\Xbf)$ in (1) with $\|h_{w_i}(\xbf_i)\|_{2,1}$ replaced by $\|h_{w_i}(\xbf_i)\|_{\varepsilon_{2,1}}$ in the second sum. As both $h_{w_i}$ and $g_{\theta}$ are compositions of Lipschitz continuous, we know the first and last terms of $\Psi_{\Theta, \gamma}^{\varepsilon}(\Xbf)$ are Lipschitz continuous. The second sum is Lipschitz continuous proved by the Lemma A2 in \cite{jimaging7110231}.
	\end{proof}
	\begin{lemma}\label{lem:bound}For any $\varepsilon > 0$,
		$\|h_{w_i}(\xbf_i)\|_{\varepsilon_{2,1}}  \leq\|h_{w_i}(\xbf_i)\|_{2,1} 
		\le  \|h_{w_i}(\xbf_i)\|_{\varepsilon_{2,1}}  + m\varepsilon$.
	\end{lemma}
	%\begin{proof} $\|h_{w_i}(\xbf_i)\|_{\varepsilon_{2,1}} = \sum\nolimits^m_{j=1}  \big(\sqrt{\| h_{w_i, j}(\xbf_i) \|^2 + \varepsilon^2}   - \varepsilon \big) = \sum\nolimits^m_{j=1}  \frac{\|h_{w_i, j}(\xbf_i) \|^2}{\sqrt{\| h_{w_i, j}(\xbf_i) \|^2 + \varepsilon^2}  + \varepsilon }
	%
	%\le \sum\nolimits^m_{j=1}  \sqrt{\| h_{w_i, j}(\xbf_i) \|^2} = \|h_{w_i}(\xbf_i)\|_{2,1}
	%\le \sum\nolimits^m_{j=1}  \sqrt{\| h_{w_i, j}(\xbf_i) \|^2 + \varepsilon^2} = \|h_{w_i}(\xbf_i)\|_{\varepsilon_{2,1}}  + m\varepsilon.
	%$
	%
	%\end{proof}
	\begin{lemma}\label{lem:inner}
		Suppose the sequence $ \{ \xt \}$ is generated by  executing Lines 3 of Algorithm 1 with fixed $ \epst = \varepsilon$ then
		\begin{enumerate}
			\item $ \| \nabla \Psi_{\Theta, \gamma}^{\varepsilon}(\xt) \| \to 0$ as $t \to \infty$.
			\item The condition in Step 4 of Algorithm 1 for reducing $\varepsilon$ can be met in finite iterations.
		\end{enumerate}
	\end{lemma}
	\begin{lemma}\label{lem:phi_decay}
		Suppose the sequence $\{ \xt\}$ is generated by Algorithm 1 with initial $\Xbf^{(0)}$, then we have 
		$ \Psi_{\Theta, \gamma}^{\epstp}(\xtp) + m \epstp \leq \Psi_{\Theta, \gamma}^{\epst}(\xt) + m \epst$ for any $t\geq 0$.
	\end{lemma}
	%
	% \begin{proof}
	% First, we can easily verify Lemma \ref{lem:bound} by simple calculation.
	% Lemma \ref{lem:inner} and \ref{lem:phi_decay} can be proved by similar way in [4,2]. %Lemma \ref{lem:bound} can be verified by simple calculation.
	% \end{proof}
	%
	The proof of \textbf{Theorem 1} is outlined below.
	%\begin{theorem}
	%Suppose that $\{\xt \}$ is the sequence generated by Algorithm \ref{alg:lda} with any initial $\Xbf^{(0)}$, $\etol=0$ and $T=\infty$. Let $ \{ \Xbf^{(t_l+1)}\}$ be the subsequence that satisfies the reduction criterion  in step \ref{reduction_cre} of Algorithm \ref{alg:lda}. Then $ \{ \Xbf^{(t_l+1)}\}$ has at least one accumulation point, and every accumulation point of $\{ \Xbf^{(t_l+1)}\}$ is a Clarke stationary point of $ \min_{\Xbf} \Psi_{\Theta, \gamma}(\Xbf)$.
	%    \label{theorem}
	%\end{theorem}
	\begin{proof}[Theorem 1]
		From Lemma \ref{lem:bound}, we have $ \Psi_{\Theta, \gamma}(\Xbf) \leq \Psi_{\Theta, \gamma}^{\varepsilon}(\Xbf) + m \varepsilon$ for any $\varepsilon>0$. Together with Lemma \ref{lem:phi_decay}, we get
		$
		\Psi_{\Theta, \gamma}(\xt) \leq \Psi_{\Theta, \gamma}^{\epst}(\xt) +m\epst\leq \cdots \leq \Psi_{\Theta, \gamma}^{\varepsilon_0}(\Xbf^{(0)}) +m\varepsilon_0 <\infty.
		$
		As $\Psi_{\Theta, \gamma}$ is coercive, we know that $\{ \xt\}$ is bounded. 
		%\blue{Let $ \Xbf^{(t_l+1)}$ denote the $l$-th $ \Xbf^{(t)}$ that satisfies the reduction criterion in step 4 of Algorithm 1. And Lemma \ref{lem:inner} indicates finite many iterations from $\Xbf^{(t_{l - 1} + 1)}$ to $ \Xbf^{(t_l)}$. So $ \{ \Xbf^{(t_l+1)} \}$ forms a subsequence of $\Xbf^{(t)}$. }
		%Accordingly, the subsequence $ \{ \Xbf^{(t_l+1)} \}$ that satisfies the reduction criterion in step 4 of Algorithm 1 is also bounded and has at least one accumulation point.
		%
		%And Lemma \ref{lem:inner} indicates finite many iterations from $\Xbf^{(t_{l} + 1)}$ to $ \Xbf^{(t_{l + 1})}$. 
		%
		%\blue{Lemma \ref{lem:inner} indicates that the reduction criterion in step 4 of Algorithm 1 will be satisfied within finitely many iterations since it was met previous time, therefore $\varepsilon_t$ will be monotonically decreasing and approaching $0$.}
		Let $ \Xbf^{(t_l+1)}$ denote the $l$-th $ \Xbf^{(t)}$ that satisfies the reduction criterion in step 4 of Algorithm 1. Then we can partition the whole sequence $\{  \Xbf^{(t)} \}$ into segments correspondingly such that the associated $\varepsilon_t = \varepsilon_{t_{l} +1} = \varepsilon_0 \eta^l$ for $t=t_{l} +1,\dots,t_{l+1} $ in the $l$-th segment, and the length of each segment is bounded according to Lemma \ref{lem:inner}.
		As $ \Xbf^{(t_l+1)}$ satisfies the reduction criterion  in step 4 of Algorithm 1, we have $ \norm{\nabla \Psi_{\Theta, \gamma}^{\varepsilon_{t_l}} (\Xbf^{(t_l+1)})} \leq \sigma  \varepsilon_{t_l} \eta= \sigma \varepsilon_{0} \eta^{l+1} \to 0$ as $l\to \infty$.  %Let $ \{ \xpp\} $ be any convergent subsequence of $\{ \xbf_{t_l+1} \}$ and denote $\epsp$ as the corresponding $\epst$ used in \mbox{Algorithm~\ref{alg:lda}} that generates $\xpp$.
		Then, there exists at least one convergent subsequence of $\Xbf^{(t_l+1)}$, dubbed $\{\Xbf^{(k + 1)}\}$, and a point $\hat{\Xbf}$ that satisfies $ \Xbf^{(k + 1)} \to \hat{\Xbf}$, $ \varepsilon_k \to 0,$ and $ \nabla \Psi_{\Theta, \gamma}^{\varepsilon_{k}}(\Xbf^{(k + 1)}) \to 0$ as $ k\to \infty$, where $\varepsilon_{k}$ is the corresponding $\varepsilon_{t_l}$ associated with $\Xbf^{(k + 1)}$. 
		Denote $\Xbf = \{\xbf_1, \xbf_2, \xbf_3 \}$. 
		The Clark subdifferential of each $\xbf_i$ is identical to \cite{jimaging7110231} expect for an additional smooth term in (1), so the analysis for each individual $\xbf_i$ is the same as \cite{jimaging7110231}.
		It has been proved in \cite{jimaging7110231} that $
		\mathrm{dist}( \nabla \Psi_{\Theta, \gamma}^{\varepsilon_{k}}(\xbf^{(k + 1)}_i), \partial^C \Psi_{\Theta, \gamma}(\hat{\xbf}_i)) \to 0,
		$ as $k \to \infty$, where $\partial^C$ denotes the Clark subdifferential. As this holds for each $\xbf_i$, then we can get it also holds for $\Xbf$ that $
		\mathrm{dist}( \nabla \Psi_{\Theta, \gamma}^{\varepsilon_{k}}(\Xbf^{(k + 1)}), \partial^C \Psi_{\Theta, \gamma}(\hat{\Xbf})) \to 0,
		$ as $k \to \infty$. Since $ \nabla \Psi_{\Theta, \gamma}^{\varepsilon_{k}}(\Xbf^{(k + 1)}) \to 0$ and $\partial^C \Psi_{\Theta, \gamma}(\hat{\Xbf})$ is closed, we conclude that $0 \in \partial^C \Psi_{\Theta, \gamma}(\hat{\Xbf})$.
	\end{proof}
	\begin{table}[htb]
		%\vspace{10pt}
		\centering
		\caption{The hyper-parameter selection. The parameters for Algorithm 1 are determined after trials by considering both algorithm convergence and the computational efficiency. The batch size is determined due to the consideration of the GPU (Nvidia GTX-1080Ti) memory and the data volume. For Algorithm 2, the selection of $\delta_{tol}$ makes the algorithm stop at around 1000 epochs. }\label{parameter}
		%\addtolength{\tabcolsep}{5pt}
		\begin{tabular}{|c|c|c|c|c|c|c|c|}
			\hline
			Initializer & Xavier  \cite{glorot2010understanding}  & $l$ & 4 & $\eta$ & 0.5 & $\etol$ & $1\times 10^{-3}$\\ 
			\hline
			Optimizer & ADAM \cite{kingma2014adam} & kernel of $h_{w_i}$ & $3 \times 3 \times 64$ & $\varepsilon_0$ & 0.001 & $\nu_\delta$ & $0.95$\\
			\hline
			Learning Rate & $0.001$ & $l'$ & 6 & $a$ & $10^5$& $\delta$ & $1 \times 10^{-3}$ \\
			\hline
			batch size & 2 & kernel of $g_{\theta}$ & $3 \times 3 \times 128$ & $\sigma$ & $ 10^3$ & $\lambda$ & $1 \times 10^{-4}$ \\
			\hline
			$\mu$ & 0.1 & $\alpha_0$ & $0.01$&  $T$ & $ 11$ & $\nu_\lambda$ & 1.001  \\
			\hline
			$\gamma$ & $1$ & $\rho_{\Theta}^{(0)}$ &0.9&  $\delta_{tol}$& $ 4.35\times 10 ^ {-6}$ & $\rho_\gamma$ & $0.9$ \\
			\hline
		\end{tabular}
	\end{table}
	\begin{figure}[!h]
		\flushleft
		\includegraphics[width=0.03\linewidth]{fig/smallwhite.png}
		\includegraphics[width=0.16\linewidth]{fig/MM.png}
		\includegraphics[width=0.16\linewidth]{fig/MM-GAN.png}
		\includegraphics[width=0.16\linewidth]{fig/MMGradAdv.png}
		\includegraphics[width=0.16\linewidth]{fig/Hi-Net.png}
		\includegraphics[width=0.16\linewidth]{fig/Proposed.png}\\
		\includegraphics[width=0.03\linewidth]{fig/FLAIR.png}
		\includegraphics[width=0.16\linewidth]{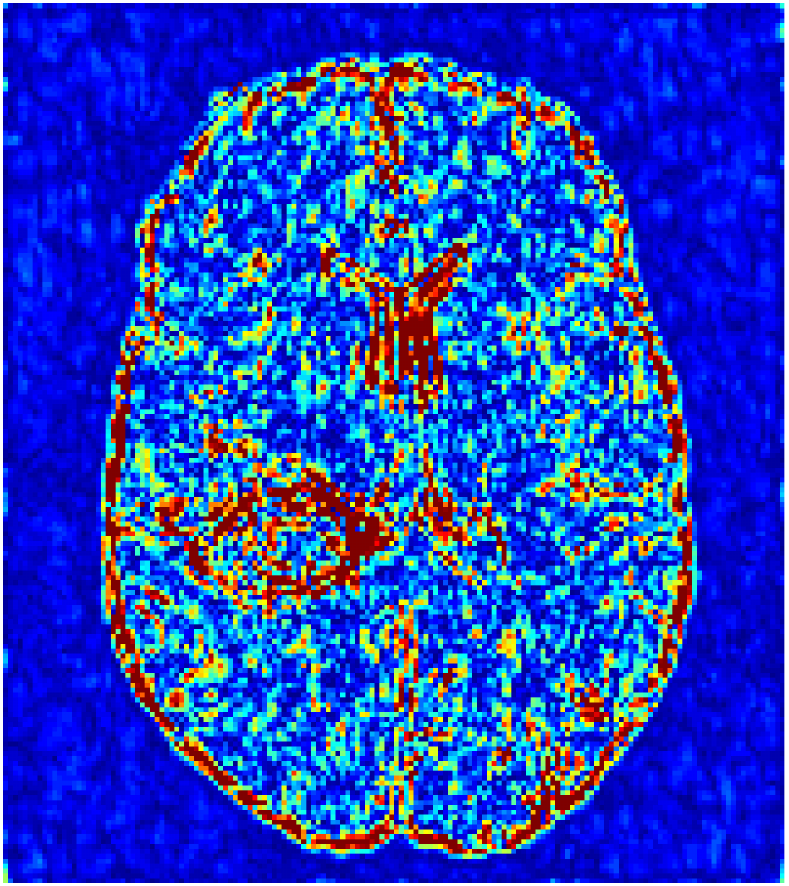}
		\includegraphics[width=0.16\linewidth]{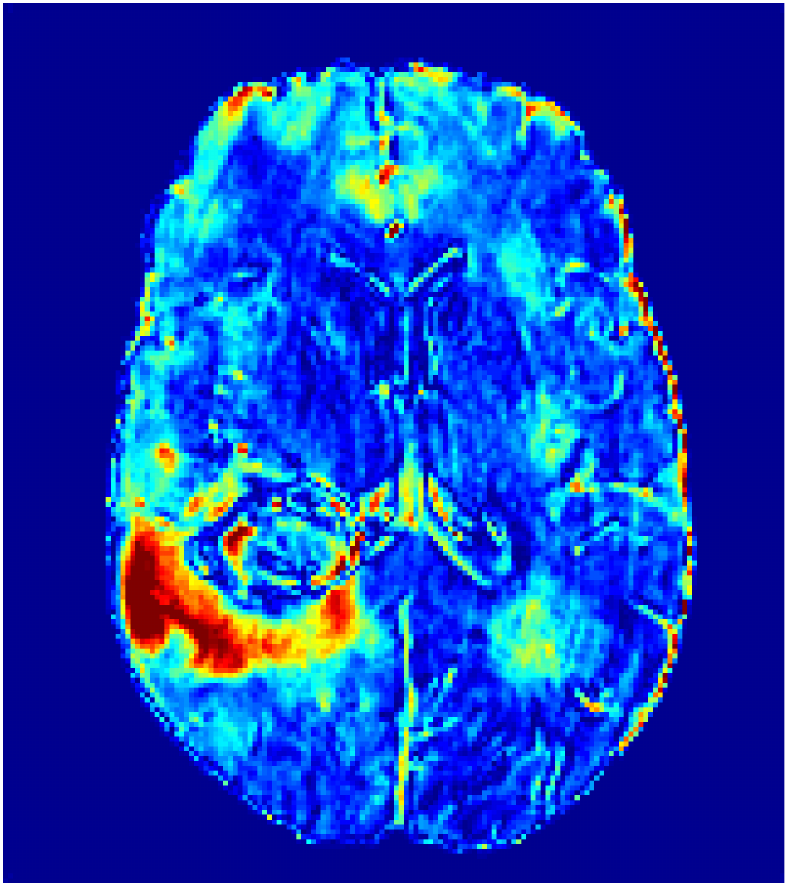}
		\includegraphics[width=0.16\linewidth]{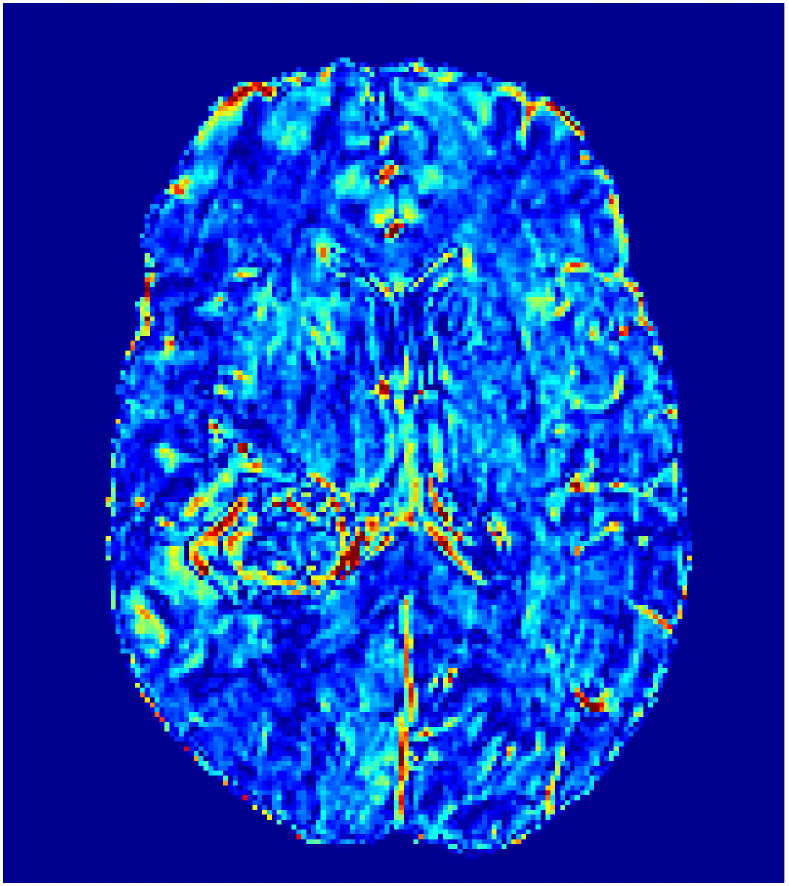}
		\includegraphics[width=0.16\linewidth]{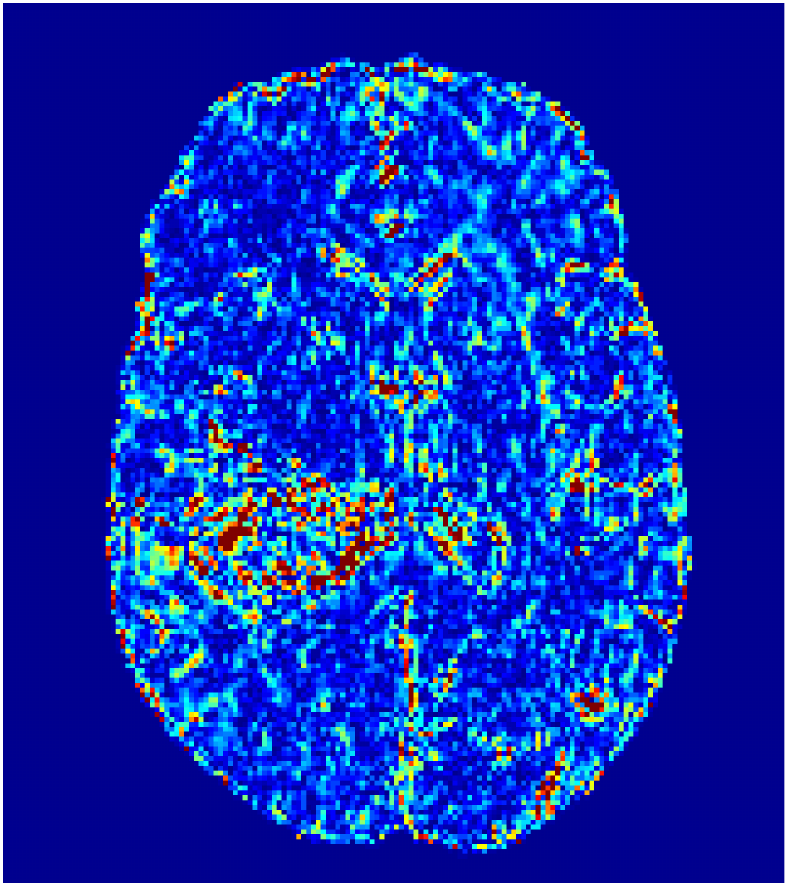}
		\includegraphics[width=0.16\linewidth]{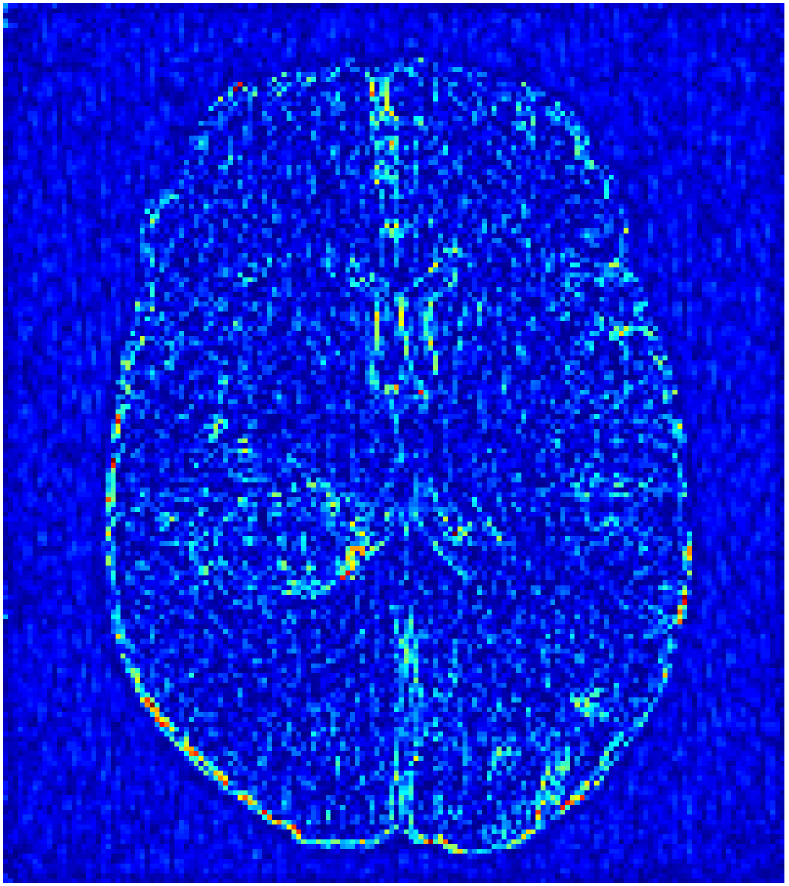}
		\includegraphics[width=0.065\linewidth]{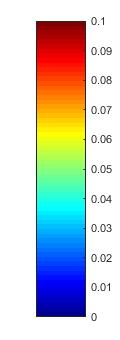}\\
		\includegraphics[width=0.03\linewidth]{fig/T2.png}
		\includegraphics[width=0.16\linewidth]{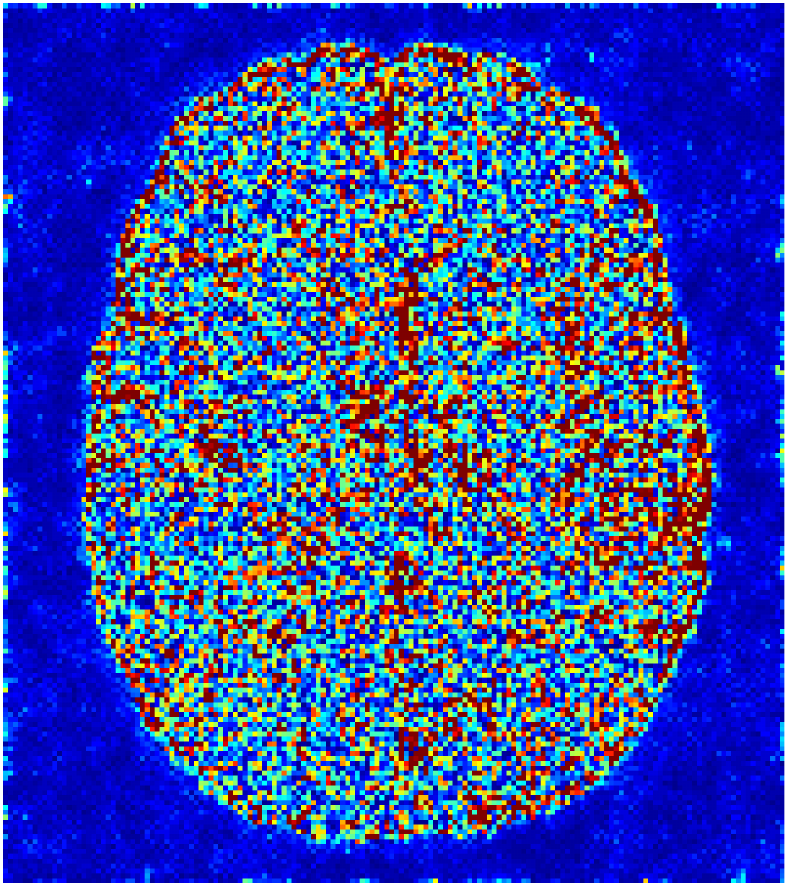}
		\includegraphics[width=0.16\linewidth]{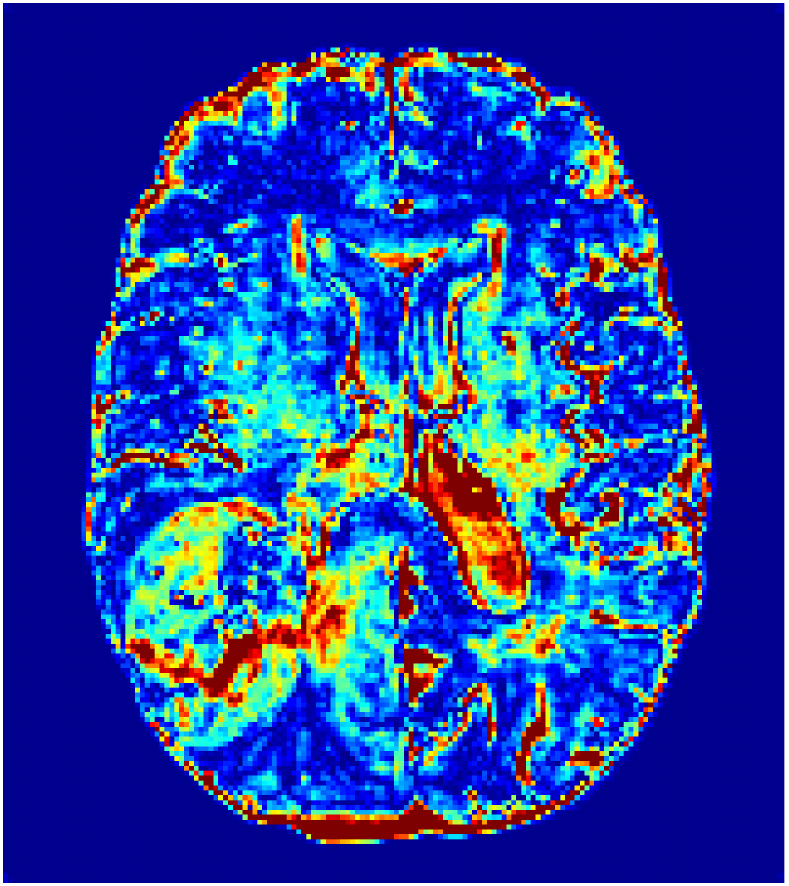}
		\includegraphics[width=0.16\linewidth]{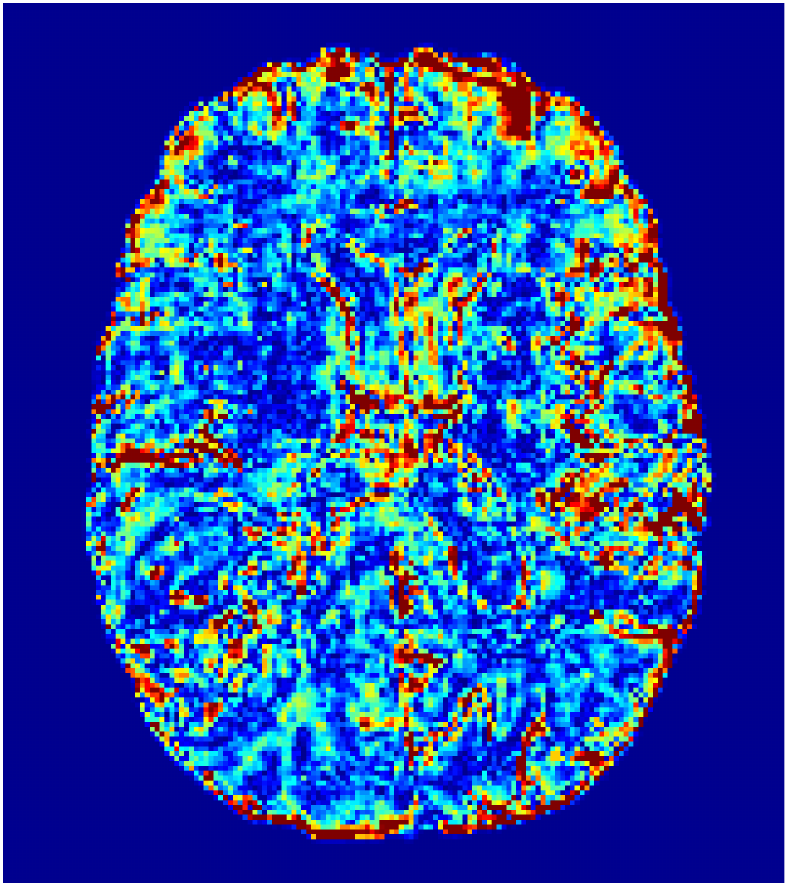}
		\includegraphics[width=0.16\linewidth]{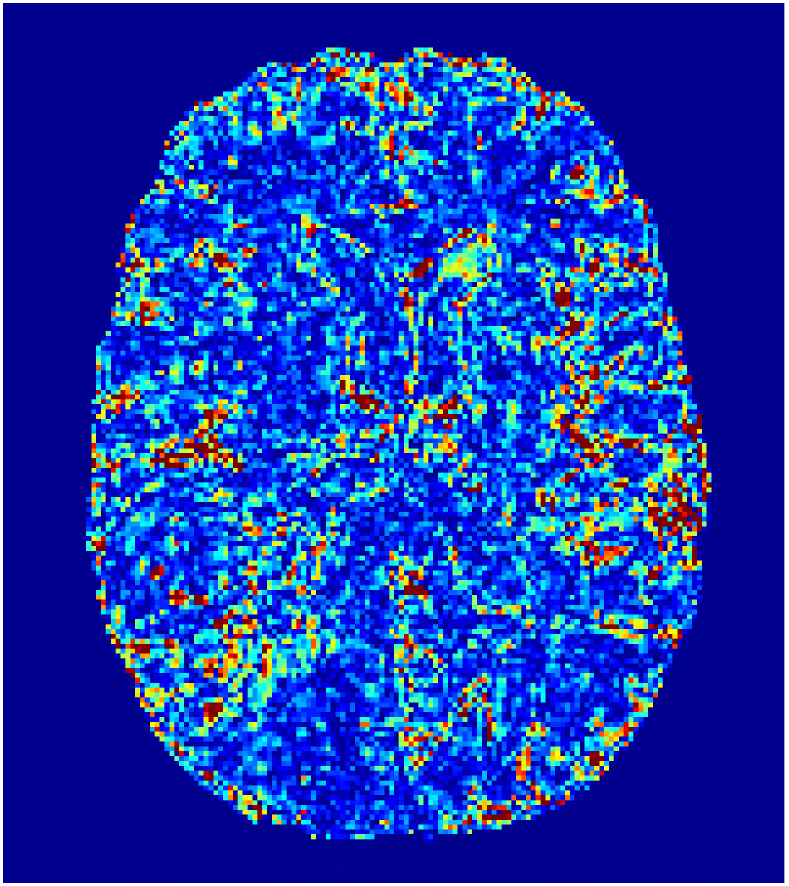}
		\includegraphics[width=0.16\linewidth]{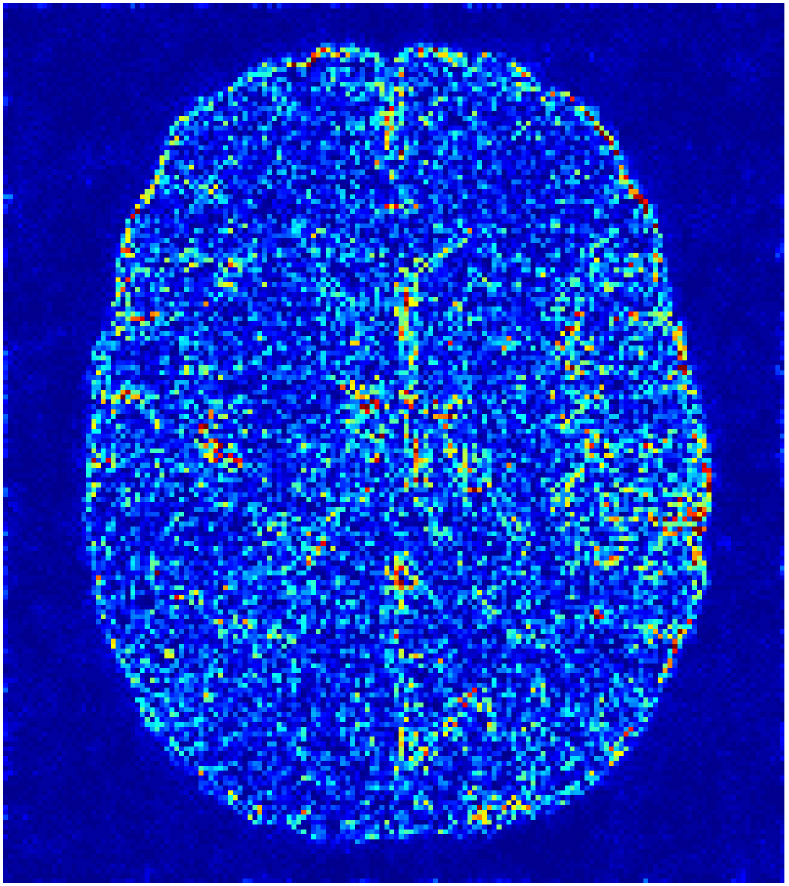}\\
		\includegraphics[width=0.03\linewidth]{fig/T1.png}
		\includegraphics[width=0.16\linewidth]{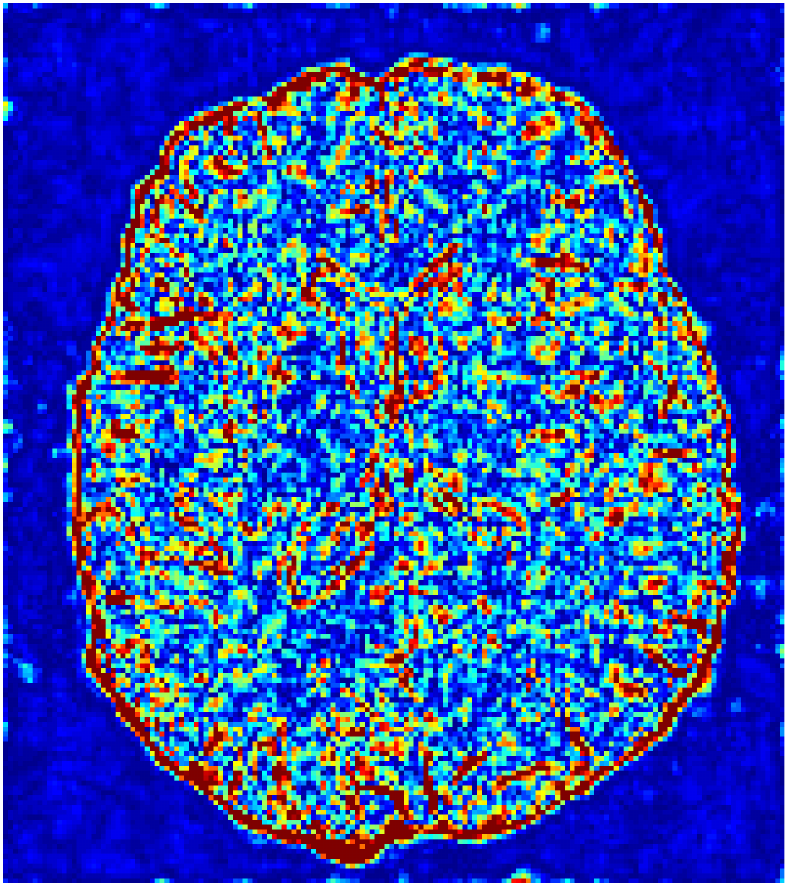}
		\includegraphics[width=0.16\linewidth]{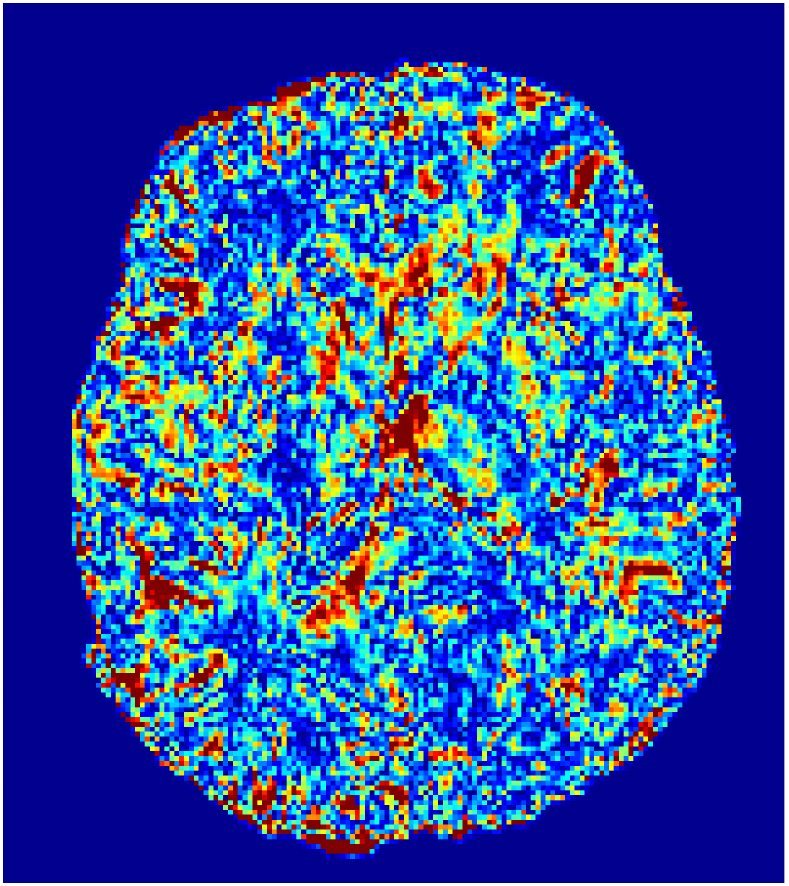}
		\includegraphics[width=0.16\linewidth]{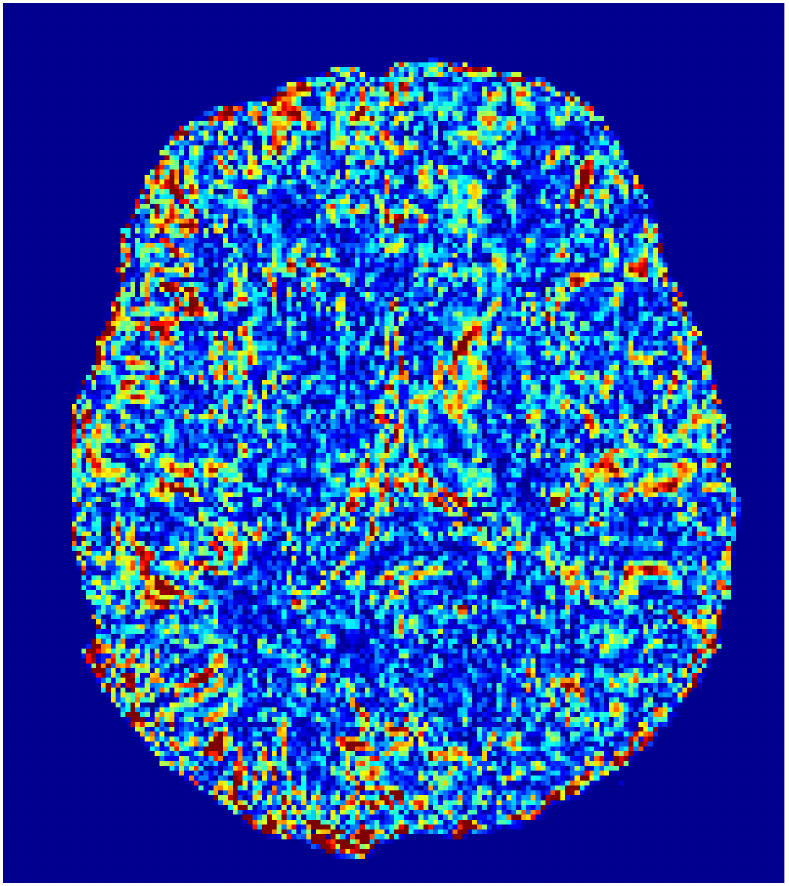}
		\includegraphics[width=0.16\linewidth]{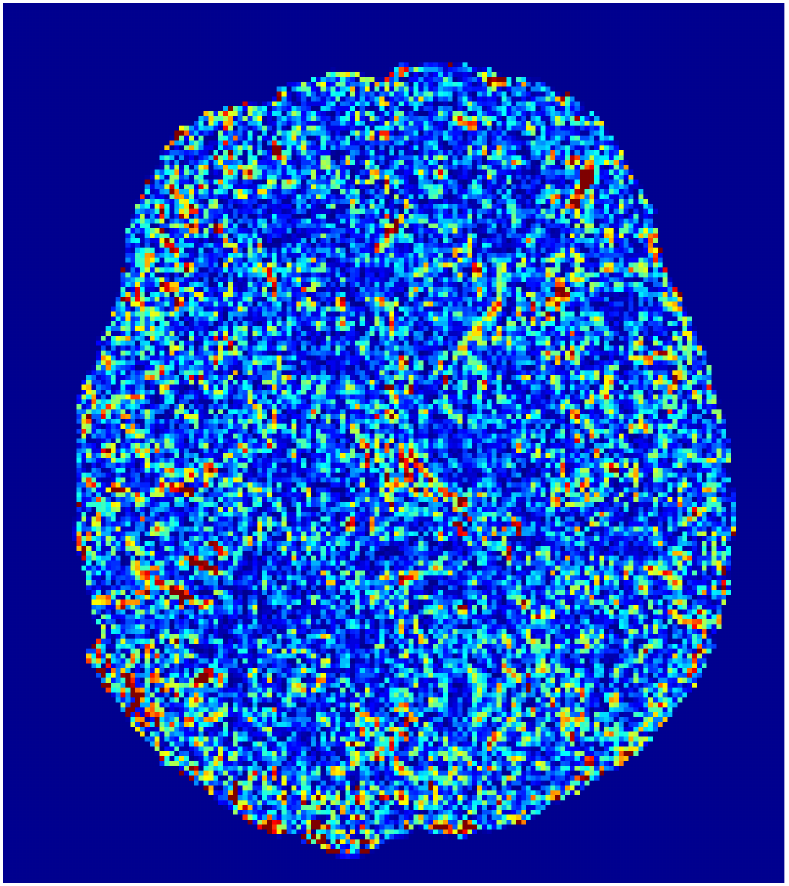}
		\includegraphics[width=0.16\linewidth]{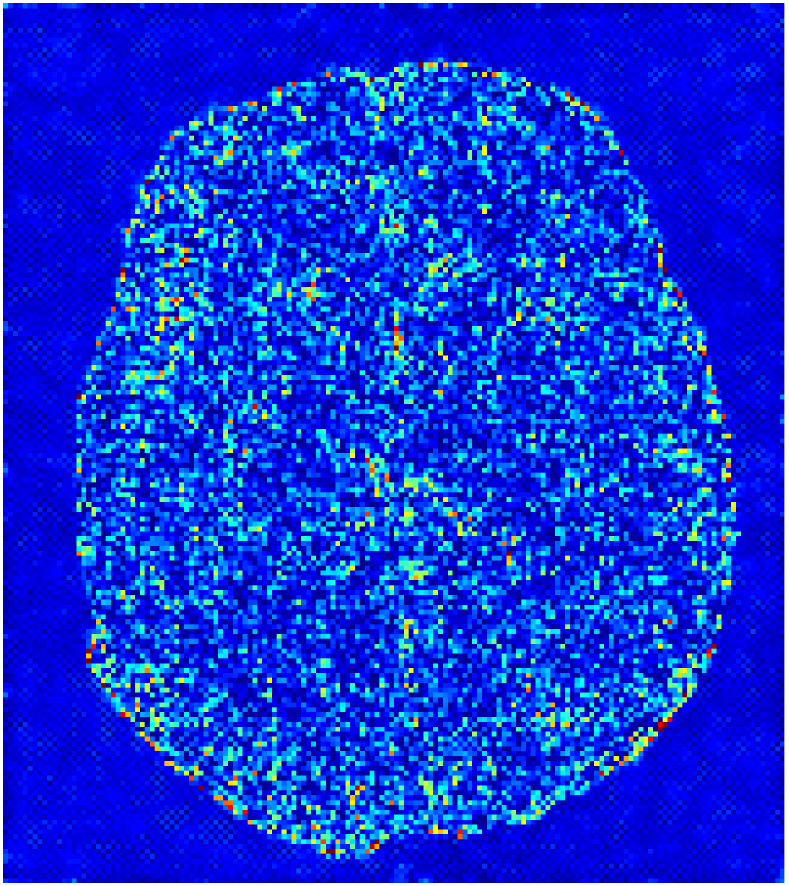}\\
		\includegraphics[width=0.03\linewidth]{fig/T1CE.png}
		\includegraphics[width=0.16\linewidth]{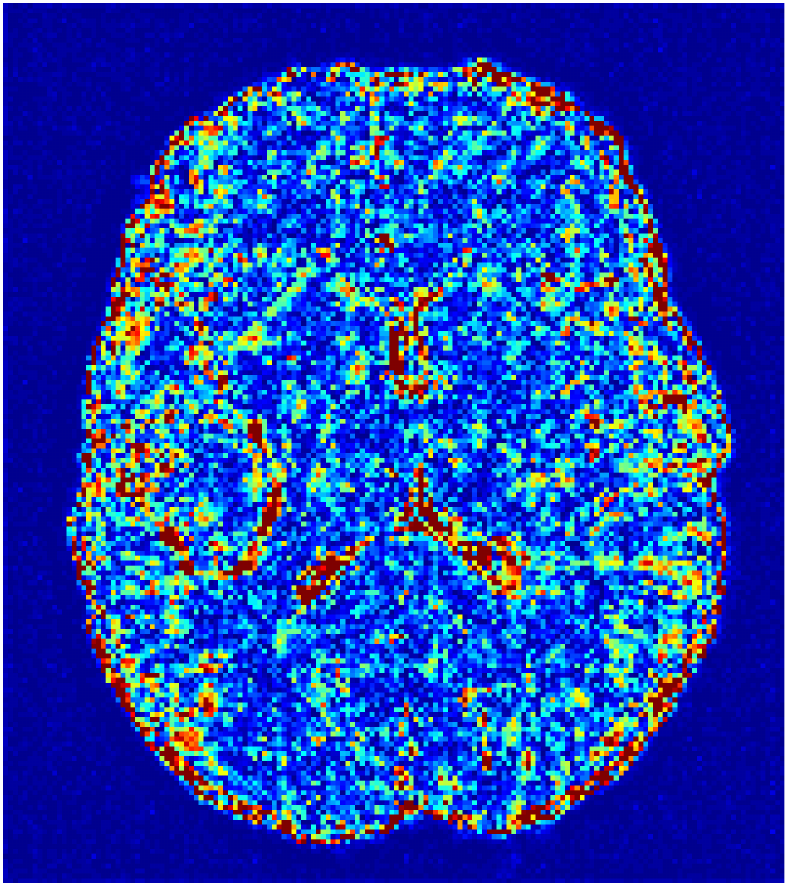}
		\includegraphics[width=0.16\linewidth]{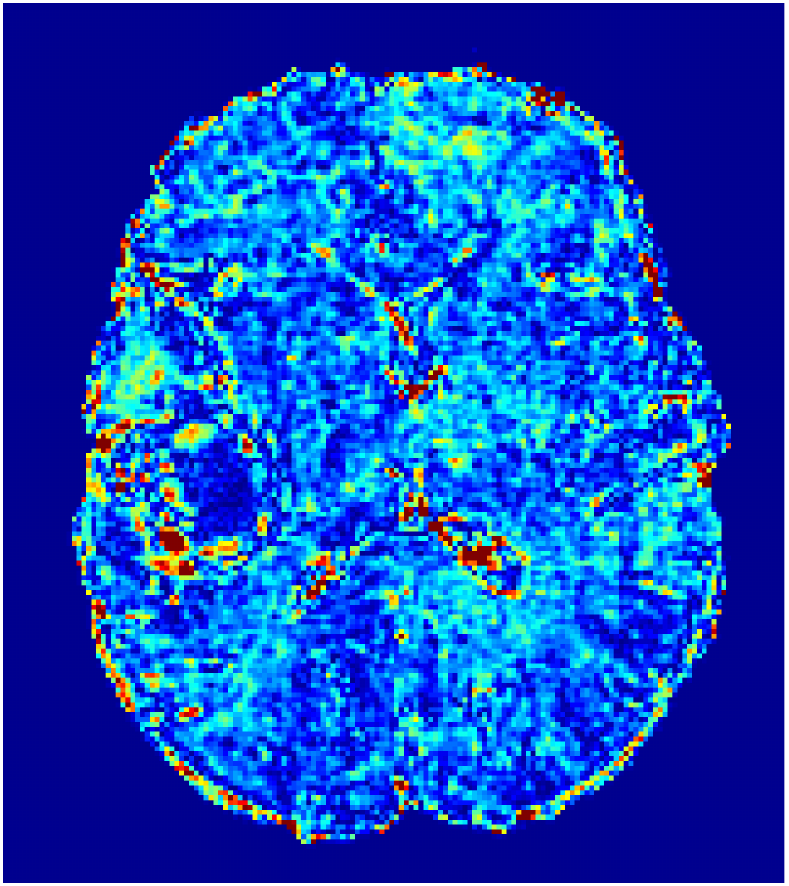}
		\includegraphics[width=0.16\linewidth]{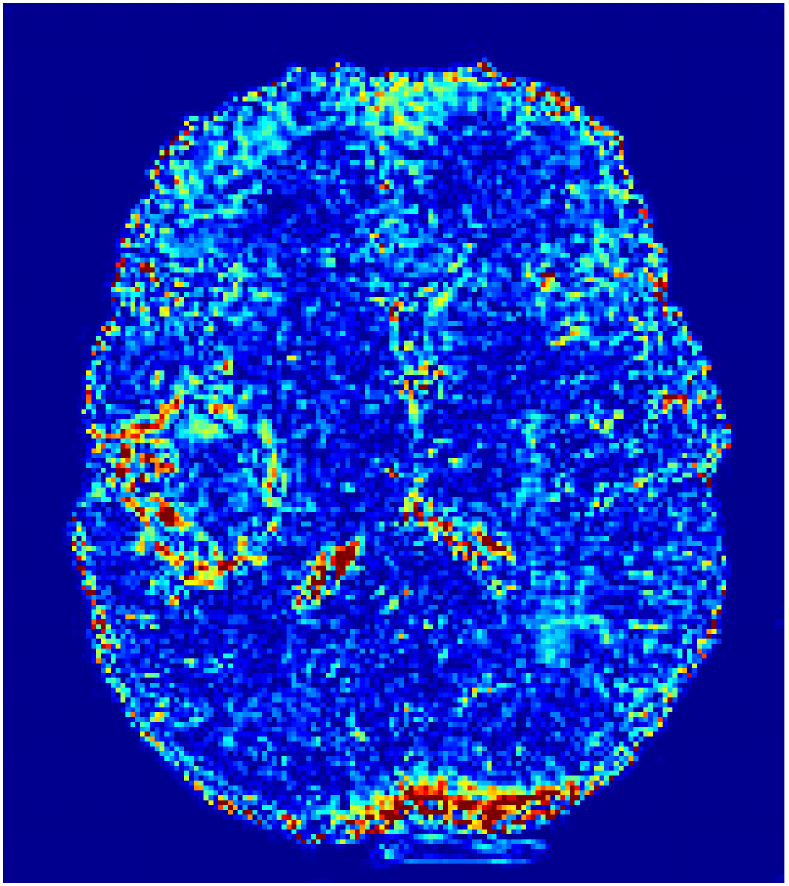}
		\includegraphics[width=0.16\linewidth]{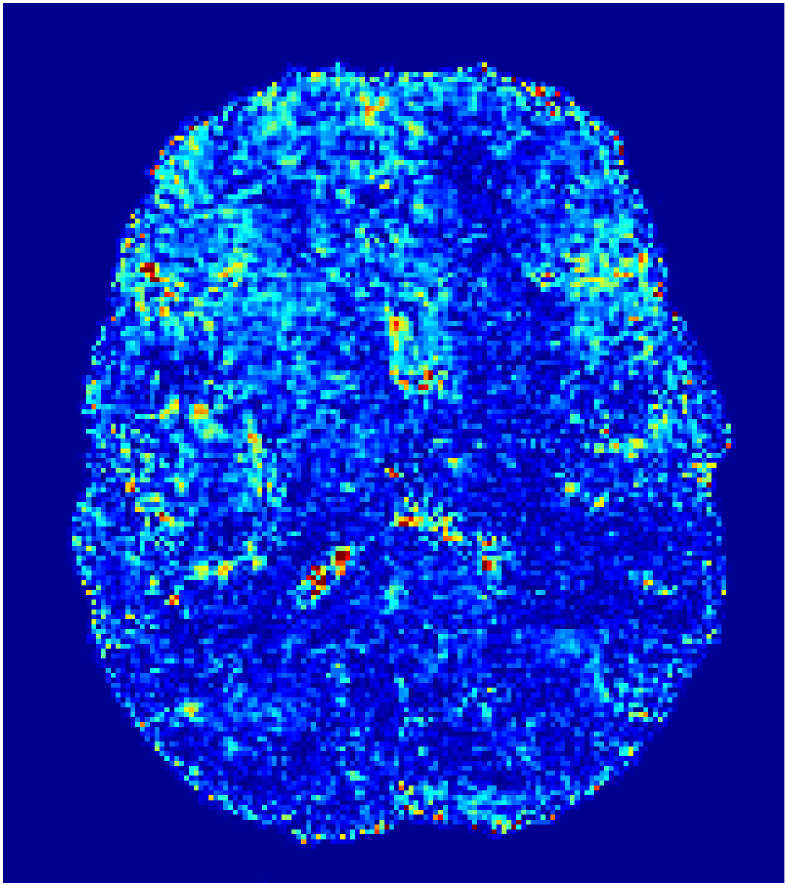}
		\includegraphics[width=0.16\linewidth]{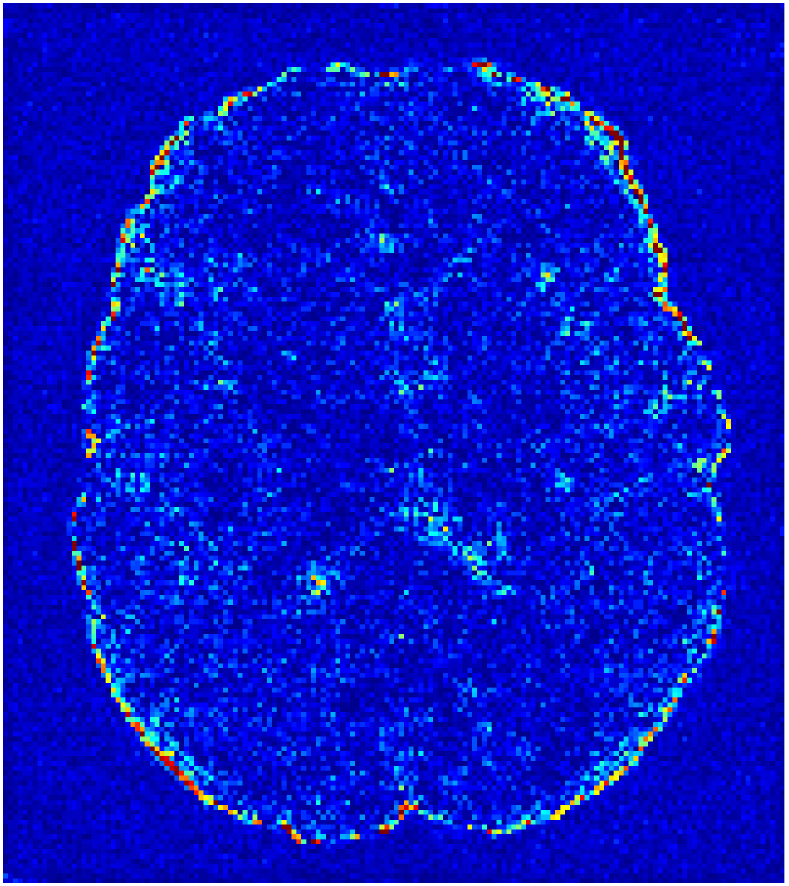}\\
		\caption{Pointwise error maps between synthetic image and the its corresponding ground truth. From first row to last row: T1 $+$ T2 $\to$ FLAIR, T1 $+$ FLAIR $\to$ T2, T2 $+$ FLAIR $\to$ T1 and T1 $+$ T2 $\to$ T1CE.  }\label{fig:synthesis_error}
	\end{figure} 

\end{document}